\documentclass[a4paper,12pt]{extarticle}
\usepackage[pagebackref=false,linktoc=page,colorlinks=true,linkcolor=blue,citecolor=blue]{hyperref}

\usepackage{url}
\usepackage{bbm}
\usepackage{bm}
\usepackage{setspace}
\usepackage{epsfig}
\usepackage{psfrag}
\usepackage[retainorgcmds]{IEEEtrantools}
\usepackage{color}
\usepackage{authblk}
\usepackage{latexsym}
\usepackage{amsmath,amssymb,amsfonts,amsthm}
\usepackage{mathrsfs}
\usepackage{mathtools}
\usepackage{caption}
\usepackage{rotating}
\usepackage[linesnumbered,lined,boxed,commentsnumbered]{algorithm2e}
\usepackage[mathscr]{eucal}
\usepackage[T1]{fontenc}
\usepackage[utf8]{inputenc}
\usepackage{graphicx}
\usepackage[numbers, authoryear]{natbib}
\usepackage{pstool,psfrag}
\usepackage{enumitem}
\usepackage{multirow}
\usepackage[all]{xy}
\usepackage[section]{placeins}
\DeclareMathAlphabet\mathbfcal{OMS}{mdugm}{b}{n}
\DeclareMathAlphabet\EuScript{U}{eus}{m}{n}
\SetMathAlphabet\EuScript{bold}{U}{eus}{b}{n}
\newcommand{\pVar}{\textsf{var}}
\newcommand{\mmd}{\mathrm{d}}
\newcommand{\md}{\,\mmd}
\newcommand{\vphi}{\varphi}
\newcommand{\borel}{\ensuremath{\mathscr{B}}}
\newcommand{\XXX}{\ensuremath{\mathcal{X}}}

\newcommand{\TTT}{\ensuremath{\mathcal{T}}}
\newcommand{\lT}{\ensuremath{\lambda_{\TTT}}}
\newcommand{\lX}{\ensuremath{\lambda_{\XXX}}}
\newcommand{\LT}{\ensuremath{\Lambda_{\TTT}}}

\newcommand{\GOm}{\ensuremath{\Gamma_{\Omega}}}
\newcommand{\GOmdot}{\ensuremath{\dot{\Gamma}_{\Omega}}}
\newcommand{\noise}{\ensuremath{\dot w}}
\newcommand{\blind}{1}

\newtheorem{proposition}{Proposition}

\makeatletter
\newcommand{\cov}{\mbox{$\text{cov}$}}

\newcommand{\PP}{\ensuremath{\mathrm{\mathbf{P}}}}
\newcommand{\RR}{\ensuremath{\mathrm{\mathbb{R}}}}
\newcommand{\SSS}{\ensuremath{\mathrm{\mathbb{S}_1}}}
\newcommand{\dd}{\ensuremath{\mathrm{d}}}

\newcommand{\NN}{\ensuremath{\mathbb{N}}}
\newcommand{\ZZ}{\ensuremath{\mathbb{Z}}}

\newcommand{\MO}{\ensuremath{\mathcal{M}_{\Omega}}}
\newcommand{\MOT}{\ensuremath{\mathcal{M}_{\Omega,\TTT}}}
\newcommand{\MA}{\ensuremath{\mathcal{M}_{\Omega\times\Theta}}}
\newcommand{\MB}{\ensuremath{\mathcal{M}_{\Omega\times\Theta,\TTT}}}

\usepackage{lipsum}

\definecolor{light-gray}{gray}{0.80} 
\title{{Bayesian inference of grid cell firing patterns using Poisson
  point process models with latent oscillatory Gaussian random fields}}
\author[1]{Ioannis Papastathopoulos
  \thanks{\small i.papastathopoulos@ed.ac.uk}}
\author[3]{Graeme Auld
  \thanks{\small graemeross.a@chula.ac.th}}
\author[1]{Finn Lindgren
  \thanks{\small finn.lindgren@ed.ac.uk}}
\author[2]{Kl\'{a}ra Zs\'{o}fia Gerlei
  \thanks{\small klara.gerlei@ed.ac.uk}}
\author[2]{Matthew F. Nolan
  \thanks{\small matt.nolan@ed.ac.uk}}
\affil[1]{{\small School of Mathematics and Maxwell Institute,
    University of Edinburgh, Edinburgh, EH9 3FD, Scotland}}
\affil[2]{{\small Centre for Discovery Brain Sciences, University of
    Edinburgh, Edinburgh, EH8 9XD, Scotland}} \affil[3] {{\small
    Department of Mathematics and Computer Science, Chulalongkorn
    University, Bangkok, Thailand}}

 \date{}

\begin{document}
\def\spacingset#1{\renewcommand{\baselinestretch}%
  {#1}\small\normalsize} \spacingset{1}
\maketitle
\begin{abstract}
  Questions about information encoded by the brain demand statistical
  frameworks for inferring relationships between neural firing and
  features of the world.\ The landmark discovery of grid cells
  demonstrates that neurons can represent spatial information through
  regularly repeating firing fields.\ However, the influence of
  covariates may be masked in current statistical models of grid cell
  activity, which by employing approaches such as discretizing,
  aggregating and smoothing, are computationally inefficient and do
  not account for the continuous nature of the physical world.\ These
  limitations motivated us to develop likelihood-based procedures for
  modelling and estimating the firing activity of grid cells
  conditionally on biologically relevant covariates.\ Our approach
  models firing activity using Poisson point processes with latent
  Gaussian effects, which accommodate persistent inhomogeneous
  spatio-directional patterns and overdispersion.\ Inference is
  performed in a fully Bayesian manner, which allows quantification of
  uncertainty.\ Applying these methods to experimental data, we
  provide evidence for temporal and local head direction effects on
  grid firing.\ Our approaches offer a novel and principled framework
  for analysis of neural representations of space.
\end{abstract}

\spacingset{1.9} 
\section{Introduction}
\label{sec:intro}

\subsection{Motivation: Inferring signals represented by neural activity}
What information does the activity of neurons in the brain convey
about the world? Questions of this kind are typically addressed by
analysing recordings of neural activity, from which spike events are
extracted, binned and smoothed
\citep{dayanTheoreticalNeuroscienceComputational2005}. However, with
such approaches, the continuous nature of the physical world is
neglected. This may mask fine scale structure in neural activity and
impede assessment of the effect of co-variables. These problems could
in principle be addressed by adopting modelling approaches that work on
the finest scale possible. Here, we develop such approaches with the
goal of analysing the firing patterns of grid cells found in the
medial entorhinal cortex.

The medial entorhinal cortex is a critical brain region for spatial
cognition and episodic memory
\citep{eichenbaumFunctionalOrganizationMedial2008,
  tukkerMicrocircuitsSpatialCoding2022,
  gerleiDeepEntorhinalCortex2021}. Neurons in this area are of
particular interest because of strong associations between their
firing activity and variables related to position and movement through
space \citep{moserSpatialRepresentationHippocampal2017}. Neurons have
been identified that encode locations
\citep{hafting2005microstructure,
  solstadRepresentationGeometricBorders2008}, orientation
\citep{sargolini2006conjunctive} and speed of movement
\citep{kropffSpeedCellsMedial2015}. The main emphasis and focus of
this paper is on the statistical modelling of the firing activity of
grid cells \citep{hafting2005microstructure}.\ Grid cells are believed
to form a coordinate system that allows spatial navigation and
learning of maps of the world. The metric of this representation is
provided by spatial firing fields that tile environments in a periodic
hexagonal pattern.\ The left panel of Figure \ref{fig:lphpp} shows a
typical data set consisting of a trajectory of an animal (grey), that
is, a record of the path followed by a moving animal over a period of
time, and an associated point pattern (black dots) of the locations at
which action potentials of a grid cell are observed.\ The action
potentials are termed \textit{firing events}. More information about
how the data were collected can be found in Section \ref{sec:data}.

\begin{figure}[htbp!]
  \centering
  \includegraphics[scale=.6]{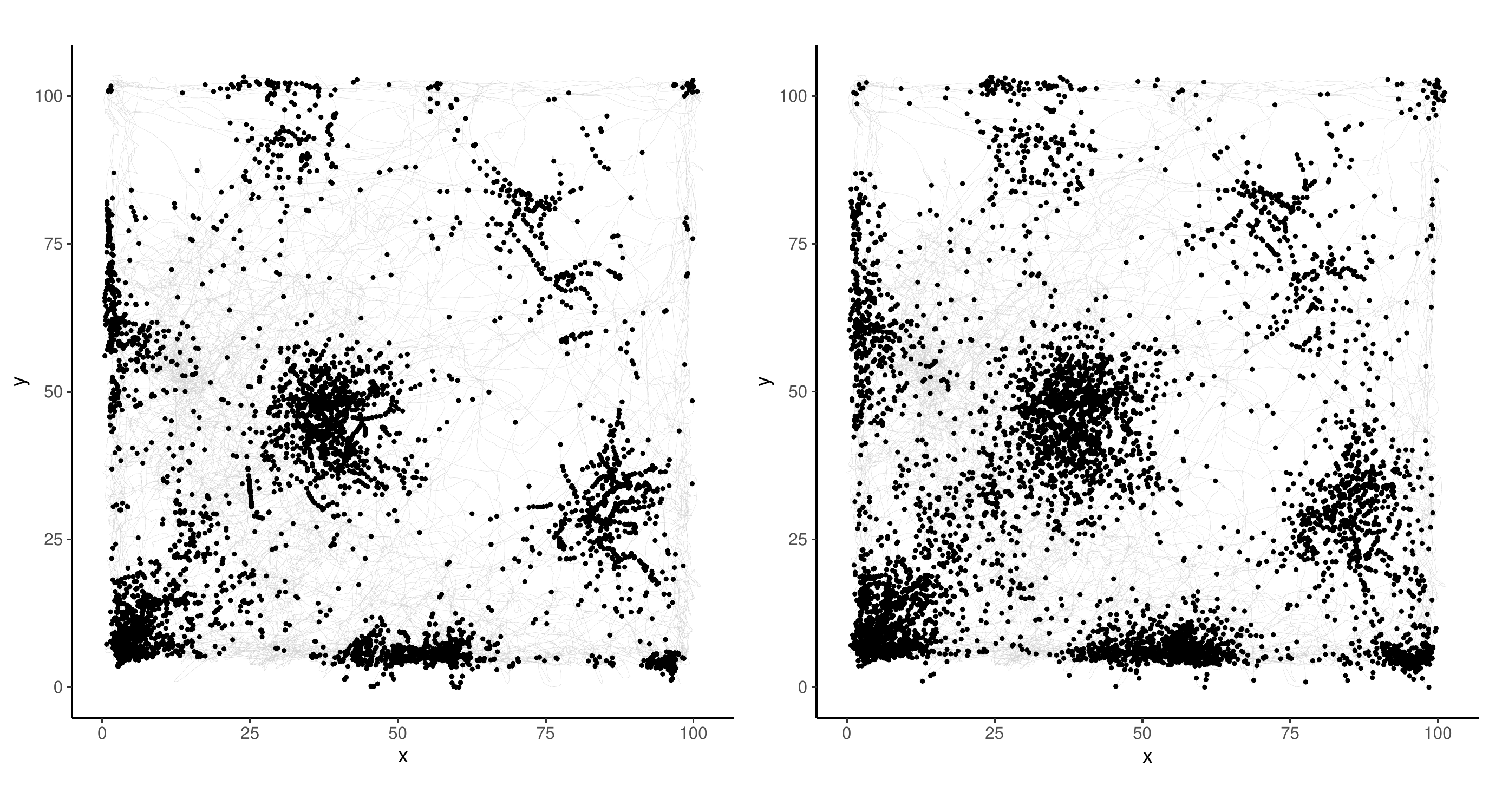}
  \caption{Left: True point pattern data (shown in black) on a trajectory $\GOm$ (shown
    in grey).\ Right: Simulated data on $\GOm$ from an inhomogeneous
    Poisson point process with estimated intensity function
    $\widehat{\lambda}_{\GOm}^0$ given by expression
    \eqref{eq:estimator_sargolini}.}
  \label{fig:lphpp}
\end{figure}

\subsection{Current approaches for inferring grid cell activity}
\label{sec:back}
Current approaches for inferring the intensity of grid firing events rely
mainly on kernel smoothing methods.\ Such methods are based on the \textit{firing rate map} which is an estimator of the rate
of occurrence of firing events per unit time at any spatial location
\citep{sargolini2006conjunctive}.\ The estimator is defined by
\begin{equation}
  \widehat{\lambda}_{\TTT}^{0}(s;h)  = \int_0^T K_h\left(\lVert \GOm(t)
    - s\lVert\right)\,N_\TTT(\dd t)\Big/ \int_{0}^T K_h\left(\lVert \GOm(t)-s
    \lVert\right)\,\dd t, ~s \in \mathbb{R}^2, T>0,  h>0,
  \label{eq:estimator_sargolini}
\end{equation}
where $\GOm\,:\,[0,T]\rightarrow \RR^2$ denotes a continuously
differentiable curve with $\GOm(t)=(s_x(t), s_y(t))$ representing the
location of the animal in $\RR^2$ at time $t\in \TTT:=[0,T]$ and $N_\TTT$
denotes the counting process of firing events over time.\ The kernel
function $K_h$ is defined by $K_h(s)=K(s/h)/h$, $s\in \RR^2$, where
$K$ denotes the probability density of the bivariate Gaussian
distribution with mean vector $(0,0)$ and variance-covariance matrix
equal to the identity matrix $I_2\in\RR^{2\times 2}$.\ The parameter
$h$ is a bandwidth parameter and $\lVert \,\cdot\, \lVert$ denotes the
Euclidean norm.\ In what follows, time is measured in seconds,
distances in centimeters and angles in radians.

To understand estimator~\eqref{eq:estimator_sargolini}, it is
helpful to decompose it into the product of two estimators, one
associated with the expected number of firing events per unit distance
travelled at location $s$, and the other associated with the expected
speed of the animal at location $s$. In particular, if we denote
\[
  W_h(t, s)=K_h(\lVert \GOm(t)-s \lVert)\Big/\int_{0}^T K_h(\lVert
  \GOm(t)-s \lVert) \, \dd t \qquad \text{and}\qquad
  \dot{\Gamma}_\Omega(t)=\left(\frac{\dd s_x(t)}{\dd t},
    \frac{ \dd s_y(t)}{\dd t}\right),
\] then we have that
$\widehat{\lambda}_{\TTT}^{0}(s;h)= \widehat{\lambda}_{\GOm}^0(s; h)
\, \widehat{v}(s;h)$, where
\begin{IEEEeqnarray*}{rCl}
  \widehat{\lambda}_{\GOm}^0(s; h) &=& \int_0^T K_h\left(\lVert
    \GOm(t) - s\lVert\right)\,N_\TTT(\dd t)\Big/ \int_{\GOm}
  K_h\left(\lVert z-s \lVert\right)\,\dd z,\\
  \widehat{v}(s;h) & = & \int_{\GOm} K_h\left(\lVert z-s
    \lVert\right)\,\dd z \Big/ \int_{0}^T K_h\left(\lVert \GOm(t)-s
    \lVert\right)\,\dd t = \int_{0}^T W_h\left(t,s\right)
  \lVert\dot{\Gamma}_\Omega(t)\rVert\,\dd t,
\end{IEEEeqnarray*}
with $\int_{\GOm}g(z)\,\dd z$ denoting the line integral of an
arbitrary integrable scalar field $g\,:\,\RR^2\rightarrow \RR$ along
the curve $\GOm$. Thus, $\widehat{\lambda}_{\TTT}^0$
decomposes into the product of a kernel-smoothing based estimator
$\widehat{\lambda}_{\GOm}^0$ for the intensity of a non-homogeneous
Poisson point process on the trajectory $\GOm$
\citep{diggle1985kernel} and a kernel-smoothing based estimator
$\widehat{v}(s;h)$ of the expected speed of the animal at a given
location $s$.\ This estimator is intrinsically
non-parametric, but it assumes that for any interval of time $I\subseteq\TTT$ with
$|I|>0$, the variation in the counts of firing events $N_\TTT(I)$ that
are observed during $I$ is solely explained by how much area the
animal explores and what regions it visits in this period of time.

Whilst seemingly simple, estimator \eqref{eq:estimator_sargolini} can
be unwieldy.\ First, kernel smoothing methods may mask fine scale
structure in the data and may produce overly smooth estimates.\
Although there are data-driven methods for choosing the bandwidth
parameter, fixing it to $h=3\text{cm}$
\citep[\textit{cf.}][supplementary material]{sargolini2006conjunctive}
appears to be common practice.\ Contrasting a realization of an
inhomogeneous Poisson point process on $\GOm$, with intensity function
equal to $\widehat{\lambda}_{\GOm}^0(, 3)$ (Figure~\ref{fig:lphpp},
right panel), with the ground-truth data (left panel), we see an
effect of oversmoothing, e.g., in the bottom left region of the square
arena.\ The ground-truth data exhibit an additional amount of
clustering which is also not captured by the inhomogeneous Poisson
point process.\ This questions the validity of the estimator and of
the generic use of a fixed bandwidth parameter.

To the best of our knowledge, no systematic study on the
smoothness properties of the firing activity of grid-cells exists.\ Although estimator \eqref{eq:estimator_sargolini} is
continuously specified, it is typically computed by
discretizing the spatial domain and then by using aggregated
statistics such as counts of firing events.\ The counts are
then used to estimate the firing activity. \ Although this simplifies statistical modelling, it comes at a
cost as it reduces estimation efficiency  and may lead to
 ecological bias, whereby inferences obtained at the coarser scale
are incorrectly carried over to the finer scale.

Second, the assumption that the variation in the number of counts
$N_\TTT(I)$ is explained only by the animal's locations in the period of time $I$ may not be
sufficient.\ For example, the firing activity may also be explained by
covariates other than the location of the animal such as the orientation of its head, which we will refer to as the
head-direction.\ By discretizing the domain of the
covariates and by computing counts within each part of the domain,
\cite{gerlei2020grid} gave evidence for an interaction effect between
the location and the head-direction, with variations in
the firing activity explained by combinations of locations and
head-directions.\ Again, the concern of potential ecological bias due to
discretization and aggregation effects also extends to the approach
adopted by \cite{gerlei2020grid}.\

In this paper we are led to develop new statistical procedures for
modelling the firing activity of grid cells that are flexible,
extendible, and coherent at the finest possible scale.\ This requires
a step-change in approach relative to current methodology and
practice.\ In particular, we adopt a principled statistical approach
based on log-Gaussian Cox processes.\ Our approach bypasses issues
related to ecological bias as it avoids discretization and
aggregation.  Additionally, it allows inclusion of covariates in a
straightforward manner.\ We build a range of statistical models which
include biologically relevant covariates. We show better predictive
performance from a model that includes an effect of interaction
between the location and the head-direction of the animal, thereby
confirming the results of \cite{gerlei2020grid}.\ Our proposed
statistical models are hierarchical and are built on the assumption
that the intensity function is random, and its prior distribution is
determined by that of a stationary Gaussian process.\ Owing to the
strong oscillatory nature of the firing activity, we use Gaussian
processes with underdamped Mat\'{e}rn covariance functions
\citep{lindetal11}.\ The upshot of this choice is that our prior
distributions mimic closely the effect of the covariates in the firing
activity, and the fitting of the models can be performed efficiently
due to the compactly supported basis function representation of our
proposed Gaussian random fields.\ As a by-product, we give closed-form
expressions for the marginal variance of Gaussian processes on $\RR$,
$\RR^2$, and $\SSS=[0,2\pi)$, with underdamped, critically damped, and
overdamped Mat\'{e}rn covariances.\ We build our models in a Bayesian
framework, which gives us a tool for fitting complex models, and the
advantage of doing this efficiently with integrated nested Laplace
approximation \citep[INLA, ][]{rue2009approximate}.\ We use the
\texttt{R} package \texttt{inlabru} \citep{inlabru}, which is based on
INLA and is specialized in dealing with structured data and has
extensive support for point process models.

\subsection{Paper structure}
In Section \ref{sec:statmodels} we describe the statistical framework
that use throughout the paper.\ Preliminary notation associated with
the experimental framework is presented Section
\ref{sec:notation_framework}.\ Section \ref{sec:pp} introduces the
class of log-Gaussian Cox processes and Section~\ref{sec:pp_models}
describes a general framework for including covariates in the
intensity function of a Poisson point process.\ Section
\ref{sec:prior_models_fields} discusses the finite-dimensional
Gaussian Markov random field (GMRF) approximations of the latent
Gaussian random fields that we use to model the effect of covariates
on the the firing rate of grid cells.\ In Section
\ref{sec:cont_random_fields} we describe the properties of the
continuous limit Gaussian random fields that the GMRFs presented in
Section \ref{sec:prior_models_fields} converge to, and derive
closed-form expressions for the variances, which are necessary for
practical model fitting and the interpretation of the models.\ Section
\ref{sec:computation} details practical information about how we
perform inference in practice.\ In particular, in Section
\ref{sec:integrals} we present approximate closed-form expressions for
the likelihood function and describe the numerical integration scheme
that we use to evaluate the intractable integral associated with the
void probability of each Poisson point process model.\ In Section
\ref{sec:prior}, we complete the specification of the Bayesian models
through choices of appropriate prior distributions for the
hyperparameters.\ Section \ref{sec:cs} contains our implementation,
analysis and interpretation of the proposed models through a case
study example from the grid cell spike train shown in Figure
\ref{fig:lphpp}.\ The data set is described in Section \ref{sec:data}.\
In Section \ref{sec:cv}, we assess the predictive performance of all
models via cross-validation and in \ref{sec:analysis} we fit and
interpret all models to the data. We conclude with a discussion that
highlights the strengths of the methodology in Section
\ref{sec:discussion}. Technical material and additional supporting
evidence is presented in \ref{sec:supplement} (Supplementary
Material).

\section{Statistical modelling}
\label{sec:statmodels}
\subsection{Notation and framework}
\label{sec:notation_framework}
The standard experimental framework relates to an animal; typically a
rodent moving in the interior of a bounded planar domain
$\Omega \subset \mathbb{R}^2$; and is modelled by a filtered
probability space $(E, \mathscr{H}, \PP, \mathscr{F})$ where
$\mathscr{F}=(\mathscr{F}_t)_{t \in \RR_+}$ denotes an increasing
family of sub-$\sigma$-algebras of $\mathscr{H}$.\ The animal moves
freely for a period of time $\TTT=[0,T]$, where
$T:E \rightarrow \RR_+$ denotes a stopping time, i.e.,
$\{T \leq t\} \in \mathscr{F}_t$ for all $t>0$.\ The stopping time $T$
marks the end of the experiment and in practice, it is usually decided
prior to the start of the experiment.\ Hence, we assume $T$ is
non-informative.\ For any $t\in \TTT$, let the two-dimensional vector
$\GOm(t)=(s_x(t), s_y(t))^\top\in\RR^2$ denote the spatial
co-ordinates of the animal in $\Omega$ at time $t$,
$\GOm(A) = \{\GOm(t)\,:\,t\in A \}$, $A\subseteq \TTT$, and
$\GOm:=\GOm(\TTT)$ for the continuously differentiable curve in
$\Omega$ representing the entire trajectory of spatial locations of
the animal throughout the duration of the experiment.\ Let
$\{s_0:=\GOm(\tau_0),\ldots,s_N:=\GOm(\tau_N)\}$ be a sample of
locations taken from $\GOm$ at approximately regular times
$\tau_0, \tau_1\ldots, \tau_N \in \TTT$, i.e.,
$\tau_{i+1}-\tau_{i}\approx \Delta \tau >0$.\ Let also
$\theta(t)\in \SSS:=[0,2\pi)$ denote the head-direction of the animal
at time $t$ and write
$\{\theta_0:=\theta(\tau_0),\dots, \theta_N:=\theta(\tau_N)\}$ for the
sample of head-directions.\ Prior to time $\tau_0=0$, electrodes that
measure the firing activity of a nerve cell are implanted in the brain
of the animal.\ Let $f:E \rightarrow \{0,1\}^{\mathcal{T}}$ be a
random variable from the sample space $E$ to the set of all functions
on $\mathcal{T}$ with image $\{0,1\}$.\ For a fixed realization of the
experiment the function $f$ encodes the activity of the neuron
according to the convention: $f(t)=1$, if the neuron is active at time
$t$ and $f(t)=0$, otherwise. In what follows, it is assumed that
available are the labeled data
$\{(s_1, \theta_1, f(\tau_1)), \ldots, (s_N, \theta_N, f(\tau_N))\}$.\
We write $t_1,\dots,t_n$ for the set of times the neuron is active,
that is, $\{t_i\,:\,i=1,\dots,n\}= \{\tau_j\,:\,
f(\tau_j)=1\}$. 
\subsection{Cox processes}
\label{sec:pp}
Let
$N_\TTT\,:\, E \times \mathscr{H} \rightarrow (\RR_+, \borel(\RR_+))$
be a random counting measure on $(\RR_+, \borel(\RR_+))$, that is,
$e\mapsto N_\TTT(e, A)$ is a random variable for each
$A\in \borel(\RR_+)$ and $A\mapsto N_\TTT(e, A)$ is, for almost every
$e\in E$, a purely atomic measure and its every atom has weight one.\
We shall denote by $N_\TTT(A)$ the former random variable and we will
associate with it the number of firing events in a Borel subset $A$ of
time.\ Let also
$\LT\,:\,E \times \mathscr{H} \rightarrow (\RR_+, \borel(\RR_+))$ be a
random measure on $(\RR_+, \borel(\RR_+))$ satisfying,
$\LT(A) = \int_A \lT(t)\,\dd t$ for any $A\in \borel(\RR_+)$, where
$\lT = (\lT(t))_{t\in \RR_+}$ is a positive continuous stochastic
process.\ In what follows, we say that the counting process $N_\TTT$
follows a Cox process with random intensity $\lambda_\TTT$ if the
conditional expectation of $\exp(-N_\TTT f)$ given the sigma-algebra
$\mathscr{F}$ generated by $\Lambda_\TTT$ is
\[
  \mathbb{E} \left\{\exp(-N_\TTT f)\mid {\mathscr{F}}\right\} =
  \exp\left(-\int_{\RR_+}[1-\exp\{-f(t)\}]\,\lambda_\TTT(t)\,\dd
    t\right), \qquad f \in \borel(\RR_+), \quad f\geq 0,
\]
that is, conditionally on the intensity function $\lambda_\TTT$, the
process $N_\TTT$ is an inhomogeneous Poisson point process with
intensity $\lambda_\TTT$ on $\TTT$, meaning that for any countable
collection of disjoint subsets
$A_1,\ldots, A_k \subseteq \borel(\RR_+)$, we have
\begin{IEEEeqnarray*}{l}
  N_\TTT(A_i)\mid \lT \sim \text{Poisson}\left(\int_{A_i} \lT(t)\,\dd
    t\right)~ \text{and}~ \text{$N_\TTT(A_1),\ldots,N_\TTT(A_k)$
    are independent given $\lT$}.
\end{IEEEeqnarray*}
Cox processes have been applied in various fields such as ecology,
epidemiology, neuroscience, and image analysis.\ This is because Cox
processes can model point patterns that exhibit both clustering and
inhomogeneity,
which is a trait observed in many natural phenomena.\ 
Another reason for the widespread use of Cox processes is their
flexibility and capacity for statistical modelling.\ The requirements
that are needed on the properties of $\lT$ so that $N_\TTT$ exists are
minimal and this accommodates a broad class of models for the random
intensity function $\lT$.\ Amongst the most popular ones, the
log-Gaussian Cox process (LGCP) models the logarithm of the random
intensity of the point process $N_\TTT$ as a Gaussian process, which
allows for flexible modelling of the variability of the point pattern,
and can model the effect of covariates in a non-linear manner.\ LGCPs
are also useful for Bayesian inference, as the Gaussian process prior
facilitates efficient numerical algorithms.


Given the $\sigma$-field $\mathscr{F}$ generated by $\LT$ and a point
pattern $\{t_1, \ldots, t_n\}$ in a bounded set $\TTT\subset \RR_+$,
the likelihood of an inhomogeneous Poisson process is 
%
\begin{equation}
  L(\lT) = \exp\left[\int_\TTT \{1-\lT(t)\}\, \dd t\right] \, \prod_{i=1}^n \lT(t_i).
  \label{eq:likelihood}
\end{equation}
Because likelihoods are needed up to a proportionality factor,
expression~\eqref{eq:likelihood} is often presented without the term
``1'' that appears in the exponent in the first factor.\ Such a
simplification in the likelihood function~\eqref{eq:likelihood} is
valid provided the domain of the point process remains
unchanged. However, when one is interested in modelling and inferring
the distribution of $N_\TTT$ subject to changes in its domain, then
\eqref{eq:likelihood} cannot be simplified further.\

\subsection{Statistical models for spike trains from grid cells}
\label{sec:pp_models}
To facilitate generality, we define the configuration space of the
system to be the set of allowable positions relative to a reference
state.\ For example, consider an animal whose body is moving in a two
dimensional Euclidean space and whose orientation in space is changing
according to a direction in the unit circle.\ Then we may take the
configuration space to be the Riemannian manifold
$\RR^2 \times \mathbb{S}_1$.\ More generally, we can label the
configuration of the animal by a system of generalized coordinates
$x=(x_1,\dots, x_m)$, where
$x_j \in \XXX_j\subseteq \RR^{d_j}, j = 1,\ldots,m$, $d_j, m\in \NN$,
so that for the aforementioned example, we have $x_1=(s_x, s_y)$ and
$x_2=\theta$ with $(s_x,s_y)$ denoting the row vector of Cartesian
coordinates in $\XXX_1 = \RR^2$ and $\theta$ denoting an angle in
$\XXX_2=\mathbb{S}_1=[0,2\pi)$.\ In the course of an experiment, the
animal undergoes dynamical evolution so it continuously changes its
position. This evolution can be specified by giving the generalized
coordinate location of the animal as a function of time, that is,
$(x_1,\dots,x_m)=(x_1(t),\dots,x_m(t))$ for $t\in\TTT$.\ This defines
a curve of covariates $\Gamma = \{\Gamma(t)\,:\,t\in\TTT\}$ where
$\Gamma(t)=(x_1(t),\dots,x_m(t))$ which we assume to be continuously
differentiable.\ It is worth to remark here that prior to the start of
the experiment, $\Gamma$ is unknown. Hence, $\Gamma$ may be taken to
be a random element, that is, a map that is measurable relative to
$\mathscr{H}$ and $\borel(\XXX)$, where
$\XXX := \prod_{j=1}^m \XXX_j$.\ In a Bayesian setting, this
assumption is equivalent to treating $\Gamma$ as a statistical
parameter in a hierarchical statistical model.\ Such an approach would
allow inference to be made about the dynamics of the trajectory of the
animal and potentially, about its effects on the intensity of the
point pattern, via the assignment of a prior distribution on $\Gamma$
and via posterior updating of prior beliefs.\ However, owing to the
fact that there is no evidence for an effect from the dynamics of
$\Gamma$ on the firing properties of grid-cells, which is the main
application and emphasis of this paper, we shall henceforth treat the
trajectory as known and fixed a priori, and we will construct
statistical models conditionally on the curve $\Gamma$.

To model the effect of the covariates on the number of spikes per unit
time at time $t$, we need to link the random mean measure $\LT$ to a
random measure on $\XXX$.\ This is accomplished by specifying
$\Lambda_\TTT$ in terms of a mapping $l$ that is measurable relative
to the join sigma-algebra $\borel(\XXX)\vee\borel(\TTT)$ and the
sigma-algebra $\borel(\RR_+)$, that is,
$\LT(\dd t) = l(\dd \Gamma(t), \dd t)$.\ To define the distance near
each point of the smooth manifold $\XXX$, we endow $\XXX$ with a
Riemannian metric that specifies infinitesimal distances on $\XXX$.\ A
Riemannian metric assigns to each $x$ a positive-definite inner
product $g_x\,:\,T_x\,\XXX\times T_x\,\XXX \rightarrow \RR$ where
$T_x\XXX$ is the tangent space at $x$. 
Given a system of generalized coordinates
$x\,:\,U\rightarrow \RR^{\sum_j d_j}$ on an open set $U$ of $\XXX$,
the vectors
$(\partial/\partial x_1\mid_x,\dots,\partial/\partial x_m\mid_x)$ form
a basis of the vector space $T_x\XXX$ for any $x\in U$. Relative to
this basis, we define the metric components of a metric tensor
$G=(g_{ij\mid x})_{1\leq i,j\leq m}$ at each $x$ by
$g_{ij\mid x} = g_x(\partial/\partial x_i \mid_x, \partial/\partial
x_j \mid_x)$.\ With this notation, our general modelling framework
presupposes that the expected number of spikes per unit time at time
$t$ equals to the expected number of spikes per unit generalized
travel distance at generalized coordinate $x(t)$, multiplied by the
generalized speed of the particle at time $t$.\ This is compactly
written as
\begin{equation}
  \Lambda_\TTT(\dd t) = \lambda_\TTT(t)\,\dd t = \lVert \dot{x}(t)\rVert \,\lambda_\XXX(x(t)),
  \label{eq:general_framework}
\end{equation}
where
$\lVert \dot{x}(t)\rVert^2 = \{(\dd x(t)/\dd t)^\top \, G\, (\dd
x(t)/\dd t)\}$ and $\lVert \dot{x}(t)\rVert$ is the local derivative
of the geodesic metric induced by the Riemannian metric $G$ along the
curve $\Gamma$ \citep{adler2007random}.\ Here $\log \lambda_\XXX$ is a
real-valued stochastic process on $\XXX$.

A natural choice for the metric tensor $G$ in our setting is
$G=\text{diag}(1,1,\delta)$ where $\delta$ is measured in meters per
radians, that is, the metric tensor only involves an axis aligned
scaling in the direction of $\theta$, which gives
$\rVert \dot{x}(t) \lVert^2= \dot{s_x}(t)^2 + \dot{s_y}(t)^2 + \delta
\, \dot{\theta}(t)^2$. The parameter $\delta$ is an unknown parameter
and in our implementations, we selected a prior distribution
degenerate at $0$, which implies that firing events cannot occur due
changes in the head-direction of the animal whilst the animal stays
still; this choice was motivated by subject expert knowledge and for
the simplicity it affords.\ Thus, in our implementation we use
$\lVert\dot{x}(t)\rVert = \lVert\dot{\Gamma}_\Omega(t)\rVert$ which is
the speed of the animal.\ However, it is worth to remark that, should
there be need for the use of a different prior for $\delta$, then this
parameter could be inferred even if instances in the data where the
animal stays still have been filtered out.

Based on these choices, in this paper we consider the models
\begin{IEEEeqnarray}{rCl}
  \MO: \,\LT \left(\mathrm{d}t\right)&=&\Lambda_{\Omega}
  \left(\GOm(\dd t)\right) := \rVert
  \GOmdot(t)\lVert\,\lambda_{\Omega}
  \left(\GOm(t)\right), \label{eq:M0}
  \nonumber\\
  \MOT: \,\LT \left(\mathrm{d}t\right)&=&\Lambda_{\Omega,\TTT}
  \left(\GOm(\dd t),\dd t\right) := \rVert
  \GOmdot(t)\lVert\,\lambda_{\Omega}
  \left(\GOm(t)\right)\,h_{\TTT}(t), \label{eq:M0T}
  \nonumber\\
  \MA: \,\LT \left(\mathrm{d}t\right)&=&\Lambda_{\Omega\times \Theta}
  \left(\GOm(\dd t), \theta(\dd t)\right) := \rVert \GOmdot(t)\lVert\,
  \lambda_{\Omega\times \Theta} \left(\GOm(t), \theta(t)\right),
  \label{eq:M1}
  \nonumber\\ 
  \MB: \,\LT \left(\mathrm{d}t\right)&=&\Lambda_{\Omega\times \Theta,
    \TTT} \left( \GOm(\dd t),\theta(\dd t), \dd t\right)\, := \rVert
  \GOmdot(t) \lVert\, \lambda_{\Omega \times \Theta} \left(\GOm(t),
    \theta(t)\right)\,h_{\TTT}(t), \label{eq:M2}\nonumber
\end{IEEEeqnarray}
where $\lambda_{\Omega}$, $\lambda_{\Omega \times \Theta}$ and
$h_{\TTT}$ denote positive continuous stochastic process on $\Omega$,
$\Omega \times \Theta$ and $\TTT$, respectively, and are described in
detail in Sections~\ref{sec:prior_models_fields} and
\ref{sec:cont_random_fields}.\

Model $\MO$ assumes that the variation in the number of spikes is
solely explained by how much area the animal explores and what regions
it visits in a period of time. In particular, in $\MO$, the expected
number of spikes per unit time at time $t$ is the product of the
expected number of spikes per travel distance at location
$\Gamma_\Omega(t)$, which is measured by
$\lambda_\Omega(\Gamma_\Omega(t))$, with the speed of the animal at
time $t$, which is measured by $\lVert\dot{\Gamma}_\Omega(t)\rVert$.\
Hence, $\MO$ is similar in spirit with the model formulation and
estimator \eqref{eq:estimator_sargolini} of
\cite{sargolini2006conjunctive} that we presented in
Section~\ref{sec:back}, and hence, it will serve as our baseline
model.\ Model $\MA$ assumes that the variation in the number of spikes
is explained by how much area the animal explores and what regions in
$\RR \times \SSS$ it visits in a period of time. Specifically, in
$\MOT$, the expected number of spikes per unit time at time $t$ is the
product of the expected number of spikes per travel distance at the
generalized coordinate $(\Gamma_\Omega(t)^\top, \theta(t))$, which is
measured by
$\lambda_{\Omega\times\Theta}(\Gamma_\Omega(t),\theta(t))$, with the
speed of the animal at time $t$.\

Models $\MOT$ and $\MB$ include an additional factor $h_{\TTT}$ which
modulates the firing rate as a function of experimental clock time,
see for example \cite{kassvent01} for a similar construction.\ Here,
the stochastic process $h_{\TTT}$ is dimensionless and hence, the
interpretation of $\lambda_{\Omega}$ and
$\lambda_{\Omega\times\Theta}$ is identical to that in models $\MO$
and $\MA$, respectively.\ The rationale for including a modulating
factor in models $\MOT$ and $\MB$ stems from the likely scenario of
additional, unobserved covariates, explaining variation in the number
and properties of firing events which is not explained by
$\lambda_\Omega$ or $\lambda_{\Omega\times\Theta}$ alone.



\subsection{Continuously specified finite-dimensional Gaussian random
  fields}
\label{sec:prior_models_fields}
Inference can be facilitated through specification of the law of the
stochastic process $\lX$.\ The most common and practically useful
choice is that of $\lX$ being a log-Gaussian process, i.e., a random
field on $\XXX$ for which the finite-dimensional distributions of
$(\log \lX(x_1), \ldots, \log \lX(x_m))$ are multivariate Gaussian for
each $1\leq m \leq \infty$ and each $(x_1, \ldots, x_m) \in \XXX$.\

The firing rate per unit displacement per unit generalized distance at
generalized coordinate $x\in\XXX$, and the modulating factor are
modelled as $\log \lX (x) = \beta + \xi_\XXX (x)$ and
$\log h_\TTT(t) = \xi_\TTT(t),$ where $\beta$ denotes a background
constant firing rate and $\xi_\XXX$ and $\xi_\TTT$ are stationary
zero-mean Gaussian processes on $\XXX$ and $\TTT$, respectively.\
Specifying a joint covariance structure over
$\XXX=\XXX_1 \times \cdots \times \XXX_m$ is highly demanding and it is
natural to simplify this task by treating the joint covariance as
separable in the input of the random field $\xi_\XXX$, i.e., to assume
the covariance function has a Kronecker product form
\citep[see][]{roug08} which asserts that
$\cov(\xi (x), \xi (x')) = \prod_{i=1}^m \zeta_i(x_i, x'_i), x, x' \in
\XXX,$,
where $\zeta_i\,:\,\XXX_i \rightarrow \RR$ is a valid covariance
function, i.e., a non-negative definite function on $\XXX_i$.\

For the differentiable manifolds $\XXX_1:=\Omega \subset \RR^2$ and
$\XXX_2:= \Theta:=[0,2\pi)$, we specify the latent random fields using
the finite dimensional models
$\xi_{\Omega}(s)= \bm \psi^{\Omega}(s)^{\top} \bm w^{\Omega}$,
$\xi_{\Theta}(\theta) = \bm \psi^{\Theta}(\theta)^{\top} \bm
w^{\Theta}$,
$\xi_{\Omega\times\Theta}(s, \theta) = \bm \psi^{\Theta}(\theta)^\top
\bm w^{\Omega \times \Theta} \bm \psi^{\Omega}(s)$ and
$\xi_{\TTT}(t) = \bm \psi^{\TTT}(t)^{\top} \bm w^\TTT$, where
$\bm w^{\Omega}_{p_{\Omega}\times 1} \sim \mathsf{N}(\bm 0, \bm
Q_{\Omega}(\bm \chi_{\Omega})^{-1})$,
$\bm w^{\Theta}_{p_{\Theta}\times 1} \sim \mathsf{N}(\bm 0, \bm
Q_{\Theta}(\bm \chi_{\Theta})^{-1})$,
$\text{vec}((\bm w_{p_{\Theta}\times p_{\Omega}}^{\Omega \times
  \Theta})^{\top})_{(p_{\Omega}\, p_{\Theta}) \times 1}\sim \mathsf{N}
(\bm 0, \bm Q_{\Omega}(\bm \chi_{\Omega})^{-1}\otimes\bm
Q_{\Theta}(\bm \chi_{\Theta})^{-1})$, and
$\bm w^{\TTT}_{p_{\TTT}\times 1} \sim \mathsf{N}(\bm 0, \bm
Q_{\TTT}(\bm \chi_{\TTT})^{-1})$,
are independent vectors of stochastic weights;
$\bm \psi^{\Omega} = (\psi_i^\Omega\,:\, i=1,\dots,p_{\Omega})$,
$\bm \psi^{\Theta}=(\psi_j^\Theta\,:\,i=1,\dots,p_{\Theta})$, and
$\bm \psi^\TTT=(\psi_k^\TTT\,:\,i=1,\dots,p_{\TTT})$, are
deterministic piecewise linear basis functions on $\RR^2, \Theta$ and
$\RR$, defined for each node on a spatial, circular, and temporal
mesh, respectively; $\text{vec}$ and $\otimes$ denote vectorization of
a matrix by stacking its column vectors on top of one another and
Kronecker product, respectively; and
$\bm Q_\Omega(\bm \chi_\Omega)$, $\bm Q_\Theta(\bm \chi_\Theta)$
and $\bm Q_\Omega(\bm \chi_\Omega)$ are positive-definite precision
matrices defined by
$\bm Q_{\cdot}(\bm \chi_{\cdot})= \tau_{\cdot}^2(\kappa_{\cdot}^4
C_{\cdot} + 2 \varphi_{\cdot} \kappa_{\cdot}^2 G_{\cdot} + G_{\cdot}
C_\cdot^{-1} G_\cdot)$, where
$\bm \chi_{\cdot}=(\kappa_\cdot, \tau_{\cdot}, \varphi_{\cdot})$,
$\kappa_\cdot, \tau_{\cdot} > 0$, $\varphi_{\cdot} \in (-1,\infty)$,
$C_{\cdot} = (C_{rc}^{\cdot})_{r,c=1}^{p_{\cdot}}$ with
$C_{rc}^{\cdot} = \langle \psi_r^{\cdot}, \psi_c^{\cdot}\rangle$, and
$G_{\cdot} = (G_{rc}^{\cdot})_{r,c=1}^{p_{\cdot}}$ with
$G_{rc}^{\cdot} = \langle \nabla \psi_r^{\cdot}, \nabla
\psi_c^{\cdot}\rangle$.\ Here, $\cdot$ is used as a placeholder for
$\Omega, \Theta$ and $\TTT$ and
$\langle\psi_r^\cdot,\psi_c^\cdot\rangle:=\int_{\cdot}
\psi_r^\cdot(e)\psi_c^\cdot(e) \lambda_{\cdot}(e)$, where
$\lambda_{\cdot}$ denotes the Lebesgue measure on the respective
domain.

The vectors of stochastic weights belong to the class of Gaussian
Markov random fields (GMRF) \citep{rueheld05} and their Markovian
properties are determined by the graph structure of the mesh.\ The
weights control the stochastic properties of
$\xi_{\Omega}, \xi_{\Theta}, \xi_{\Omega \times \Theta}$ and
$\xi_\TTT$, and are chosen so that the distibution of each linear
combination converges to the distribution of the solution of an SPDE,
see \cite{simpsetal16} for more
details.
\subsection{Limiting infinite-dimensional Gaussian random fields}
\label{sec:cont_random_fields}

When the parameters $\varphi_{\Omega}$ and $\varphi_{\Theta}$ are
equal to unity, then all random fields presented in
Section~\ref{sec:prior_models_fields} converge strongly
\citep{lindetal11, simpsetal16} to Gaussian random fields with
Mat\'{e}rn covariance functions, a class of random fields that is
widely adopted in spatial statistics \citep{stein1999interpolation}.\
In particular, the models that we presented in the previous section
are Hilbert space projections of solutions of the Whittle--Mat\'{e}rn
stochastic (partial) differential equation \citep{lindetal11,
  lindgren2022spde}.\ The continuously specified infinite-dimensional
limits of the Gaussian random fields $\xi_{\Omega}$, $\xi_{\Theta}$
and $\xi_{\TTT}$ are Gaussian random fields with power spectral
mass/density functions
\begin{IEEEeqnarray}{rCl}
  f_\Omega(\bm \omega) &=&
  \frac{\sigma_{\Omega}^2}{(2\pi)^{2}}\left(\frac{1
    }{\kappa_{\Omega}^4 + 2\, \varphi_{\Omega}\, \kappa_\Omega^2
      \,\lVert\bm \omega\rVert^2 +
      \lVert\bm \omega\rVert^4}\right)^{\alpha_\Omega/2}, \qquad \bm \omega \in \RR^2\label{eq:ps_omega}\\
  f_\Theta(\omega)&=& \frac{\sigma_{\Theta}^2}{2\pi}
  \left(\frac{1}{\kappa_{\Theta}^4+2\varphi_{\Theta}\kappa_{\Theta}^2\omega^2+
      \omega^4}\right)^{\alpha_\Theta/2}, \qquad \omega \in \ZZ\label{eq:ps_theta}\\
  f_\TTT(\omega) &=& \frac{\sigma_\TTT^2}{2\pi}
  \left(\frac{1}{\kappa_\TTT^4+2\varphi_\TTT\kappa_\TTT^2\omega^2+\omega^4}\right)^{\alpha_\TTT/2},\qquad
  \omega \in \RR. \label{eq:ps_T}
\end{IEEEeqnarray}
respectively, where the shape parameters satisfy
$\alpha_\Omega = \alpha_\Theta = \alpha_\TTT = 2$.\ 
We work with the special case where all shape parameters are equal to
2, for the simplicity and clarity this choice affords.
The covariance functions may be given in analytic form for some cases
only, e.g., on $\RR$ and $\RR^2$, inversion of the spectrum gives
\citep{lindetal11}
\begin{IEEEeqnarray*}{rClr} \zeta_\Omega(s) &=&
  \frac{\sigma_\Omega^2}{4 \pi \sin(\pi \gamma_\Omega) \kappa_\Omega^2
    i} \left[K_0\left\{\kappa_\Omega \lVert s \rVert
      \exp\left(-\frac{i \pi \gamma_\Omega}{2}\right)\right\}
    -K_0\left\{\kappa_\Omega \lVert s \rVert \exp\left(\frac{i \pi
          \gamma_\Omega}{2}\right)\right\}\right],& \\
  \zeta_\TTT(t) &=& \frac{\sigma_\TTT^2}{2 \sin(\pi \gamma_\TTT)
    \kappa_\TTT^3} \exp\left\{-\kappa_\TTT\cos\left(\frac{\pi
        \gamma_\TTT}{2} |t|\right)\right\} \sin\left\{\frac{\pi
      \gamma_\TTT}{2} + \kappa_\TTT\sin\left(\frac{\pi
        \gamma_\TTT}{2}|t|\right)\right\}, & \quad
\end{IEEEeqnarray*} where $\varphi_\Omega, \varphi_\TTT\in (-1,1]$ and
$\gamma_{\cdot} = \text{arccos}(\varphi_{\cdot})/\pi$.\ On $\SSS$ the
covariance is obtained in analytic form only for the case
$\varphi_\Theta = 1$ and is given in Proposition
\ref{prop:covariance_circular} below.
\begin{figure}[htbp!]
  \centering
  \includegraphics[trim= 0 20 0 10]{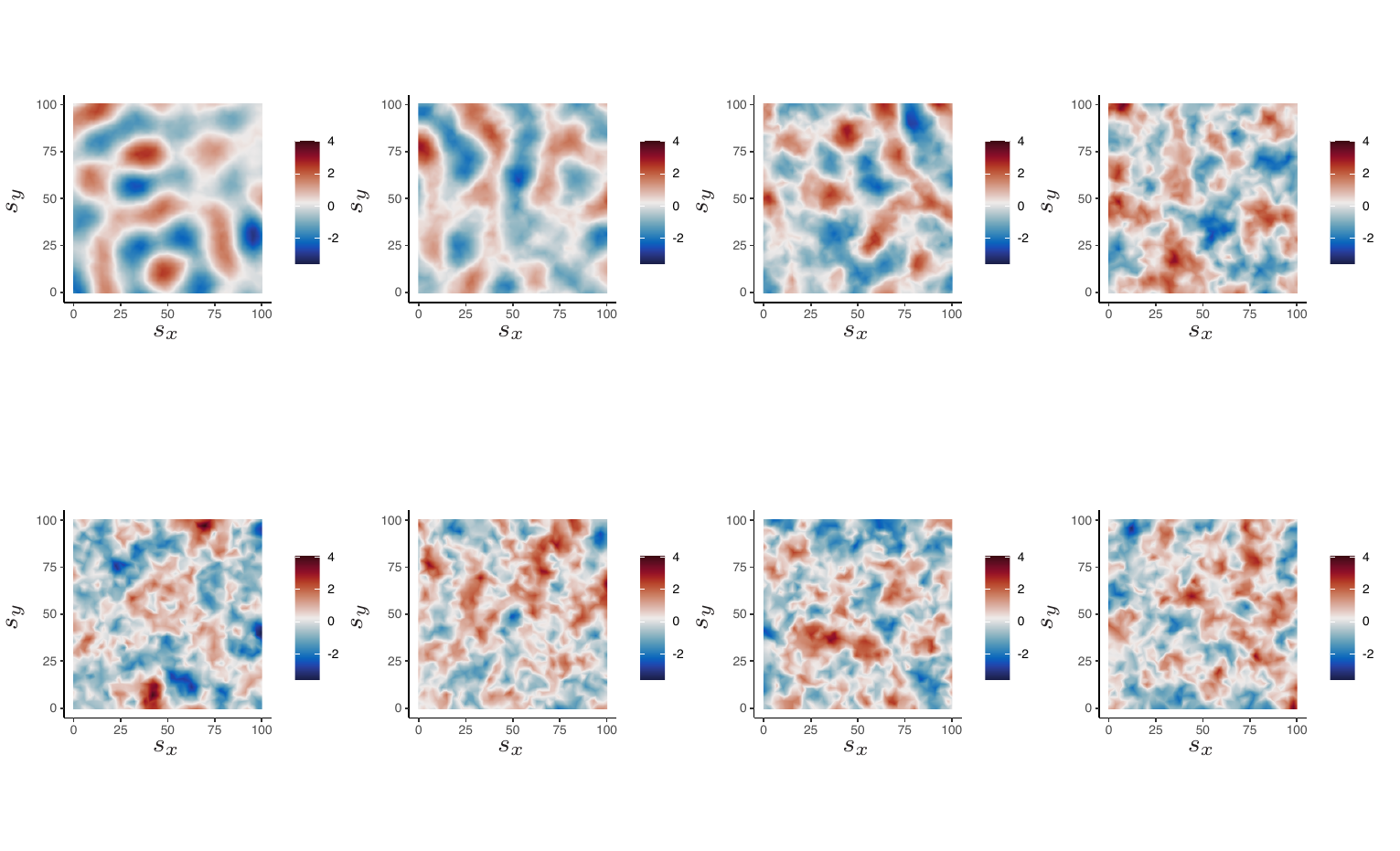}  
  \caption{Realizations of a zero-mean Gaussian random field
    $\xi_\Omega(s)$ with spectral density \eqref{eq:ps_omega}, where
    $\alpha_\Omega=2$, $\kappa_\Omega = 0.2$ and
    $s_\Omega^2=1$.\ \textit{Top, left to right}: shown for
    $\varphi_\Omega=-9.99/10, -9.9/10, -9/10, -2/10$. \textit{Bottom,
      right to left}: shown for
    $\varphi_\Omega=2/10, 9/10, 9.9/10, 9.99/10$.}
  \label{fig:osc_fields}
\end{figure}
\begin{proposition}
  \label{prop:covariance_circular}
  A stationary Gaussian process $\xi_\Theta(\theta)$, $\theta\in \SSS$,
  with power spectrum \eqref{eq:ps_theta} with $\varphi_{\Theta}=1$, has
  covariance function
\begin{IEEEeqnarray*}{rCl}
  \zeta_\Theta(\theta) &=& \frac{\sigma_\Theta^2}{4 \sinh^2(\pi
    \kappa_\Theta) \kappa_\Theta^2}\left\{\frac{|\theta|}{2}\,\cosh
    \{(2\pi -|\theta|)\kappa_\Theta\}+ \right. \\ \\
  && \qquad \qquad  \left.+ \left(\pi -
      \frac{|\theta|}{2}\right)\,\cosh(|\theta| \kappa_\Theta) +
    \kappa_{\Theta}^{-1} \left[\cosh\{(\pi-|\theta|)\kappa_\Theta)\,
      \sinh(\pi \kappa_\Theta\}\right]\right\}.
\end{IEEEeqnarray*}
\end{proposition}
A proof is given in Section \ref{app:covariance_circular} of the
supplement.\ Although the covariance is available in closed-form only
under special cases, this is not a substantial inconvenience, as what
is needed for practical model fitting are the discretised precision
matrices from Section~\ref{sec:prior_models_fields}, and for
convenience, e.g.\ for normalisation and parameter interpretation, the
marginal variance. The marginal variances for the processes on
$\RR^2$, $\SSS$ and $\RR$ are given in Propositions
\ref{prop:variance-on-R2}, \ref{prop:variance-on-S1} and
\ref{prop:variance-on-R}.
\begin{figure}[htbp!]  \centering
  \includegraphics[scale=.45, trim= 0 40 0 20]{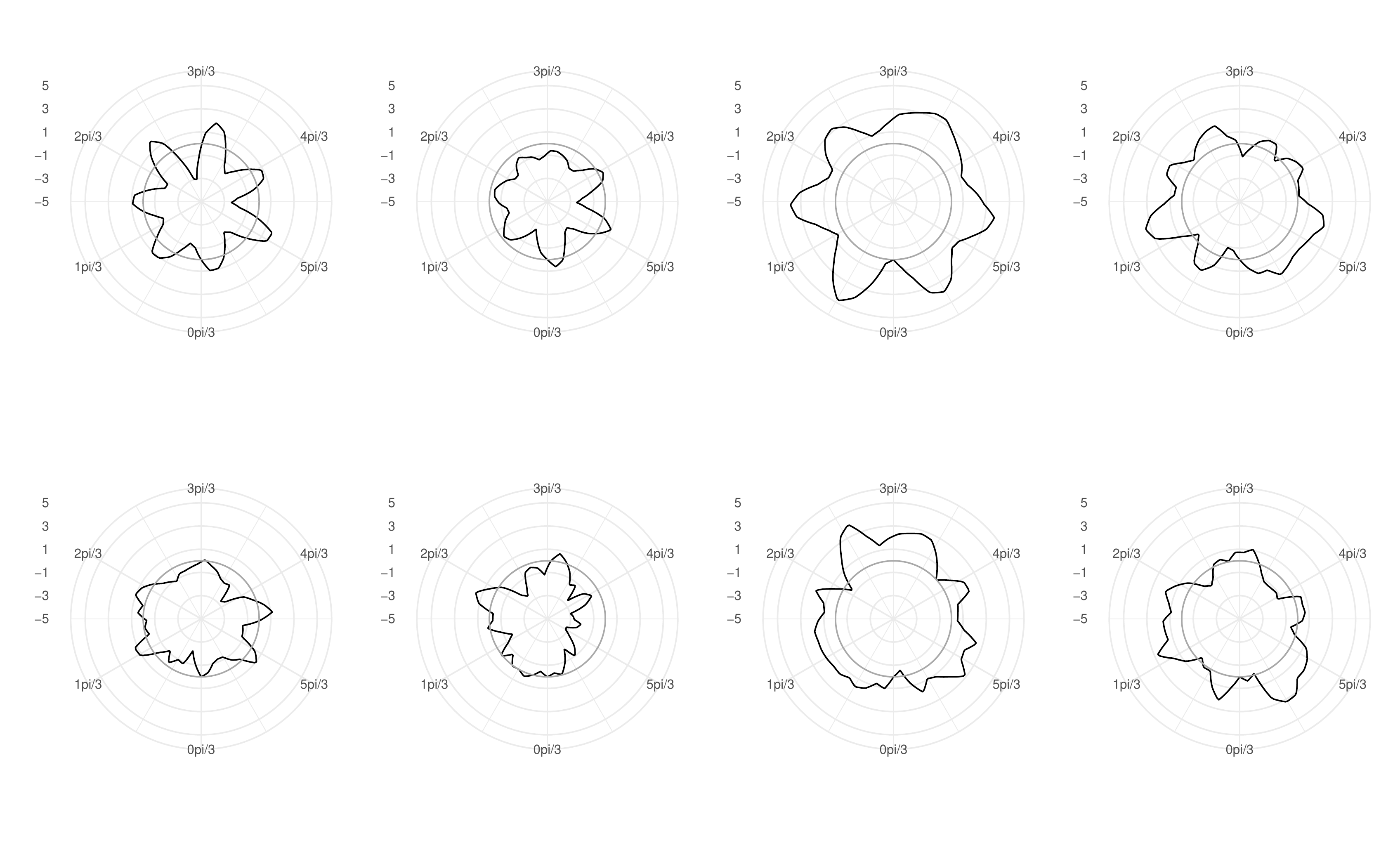}
  \caption{Realizations of $\xi_\Theta(\theta) - \xi_\Theta(0)$ where
    $\xi_\Theta$ is a zero-mean Gaussian random field with spectral
    mass function \eqref{eq:ps_theta}, where $\alpha_\TTT=2$,
    $\kappa_\Theta = 2\pi$ and $s_{\Theta}^2=1$. Realizations are
    plotted in polar coordinates, with scale units marked by circles
    (center refers to scale value $-5$ and outer circle refers to
    scale value $5$).\ The circle shown in dark grey refers to the
    scale value $0$. \textit{Top, left to right}: shown for
    $\varphi_\Theta=-9.99/10, -9.9/10, -9/10, -2/10$. \textit{Bottom,
      right to left}: shown for
    $\varphi_\Theta=2/10, 9/10, 9.9/10, 9.99/10$. }
  \label{fig:circ_fields}
\end{figure}
\begin{proposition}\label{prop:variance-on-R2}
  A stationary Gaussian process $\xi_\Omega(s)$, $s\in \RR^2$, with
  power spectrum \eqref{eq:ps_omega} has marginal variance
  $s_\Omega^2:=\text{var}\{\xi_\Omega(s)\}$ given by\normalfont
  \begin{equation*}
    s_\Omega^2= 
    \begin{dcases}
      \frac{\sigma_{\Omega}^2}{4\pi
      \kappa_\Omega^2 \sqrt{1-\varphi_\Omega^2}} \arccos \varphi_\Omega, & \text{underdamped}, \varphi_\Omega \in (-1,1),\\
      \frac{\sigma_\Omega^2}{4\pi\kappa_\Omega^2},& \text{critically damped}, \varphi_\Omega=1,\\
      \frac{\sigma_\Omega^2}{4\pi\kappa_\Omega^2
      \sqrt{\varphi_\Omega^2-1}}\left\{\frac{\pi}{2} -
      \text{atan}\left( \frac{\varphi_\Omega}{\sqrt{\varphi_\Omega^2-1}}\right)\right\},
          & \text{overdamped}, \varphi_{\Omega} \in (1,\infty).
    \end{dcases}
  \end{equation*}
\end{proposition}
\newpage
\begin{proposition}\label{prop:variance-on-S1}
  A stationary Gaussian process $\xi_\Theta(\theta)$, $\theta\in\SSS$,
  with power spectrum \eqref{eq:ps_theta} has marginal variance
  $s_\Theta^2:=\text{var}\{\xi_{\Theta} (\theta)\}$ given by\normalfont
  \[
    s_\Theta^2 =
    \begin{dcases}
      \frac{\sigma_{\Theta}^2}{4 \kappa_{\Theta}^3\sqrt{1-\phi_\Theta^2}} \frac{2\{y \sin(2\pi \kappa_\Theta x) + x \sinh(2\pi \kappa_\Theta y)\}}{\left\{\cos(2\pi \kappa_\Theta x)-\cosh(2\pi \kappa_\Theta y) \right\}}&\text{underdamped}, \varphi_{\Theta} \in (-1,1)\\
      \frac{\sigma_{\Theta}^2}{4 \kappa_{\Theta}^3}& \text{critically damped}, \varphi_{\Theta} = 1\\
      \frac{\sigma_\Theta^2 }{4 \kappa_\Theta^3 \sqrt{\varphi_\Theta^2-1}}
      \left(\frac{\cot \left(i z\right)} {i z}-\frac{\cot \left(i
      \overline{z} \right)} {i \overline{z}} \right),\quad & \text{overdamped}, \varphi_{\Theta} \in (1,\infty)
    \end{dcases}
  \]
  where $x=\cos(\gamma/2)$, $y=\sin(\gamma/2)$ and
  $\gamma\in[-\pi, 0)$ satisfies
  $\exp(\iota \gamma)=-\phi_\Theta- \sqrt{1-\phi_\Theta^2}$.
\end{proposition}
\begin{figure}[htbp!]  \centering
  \includegraphics{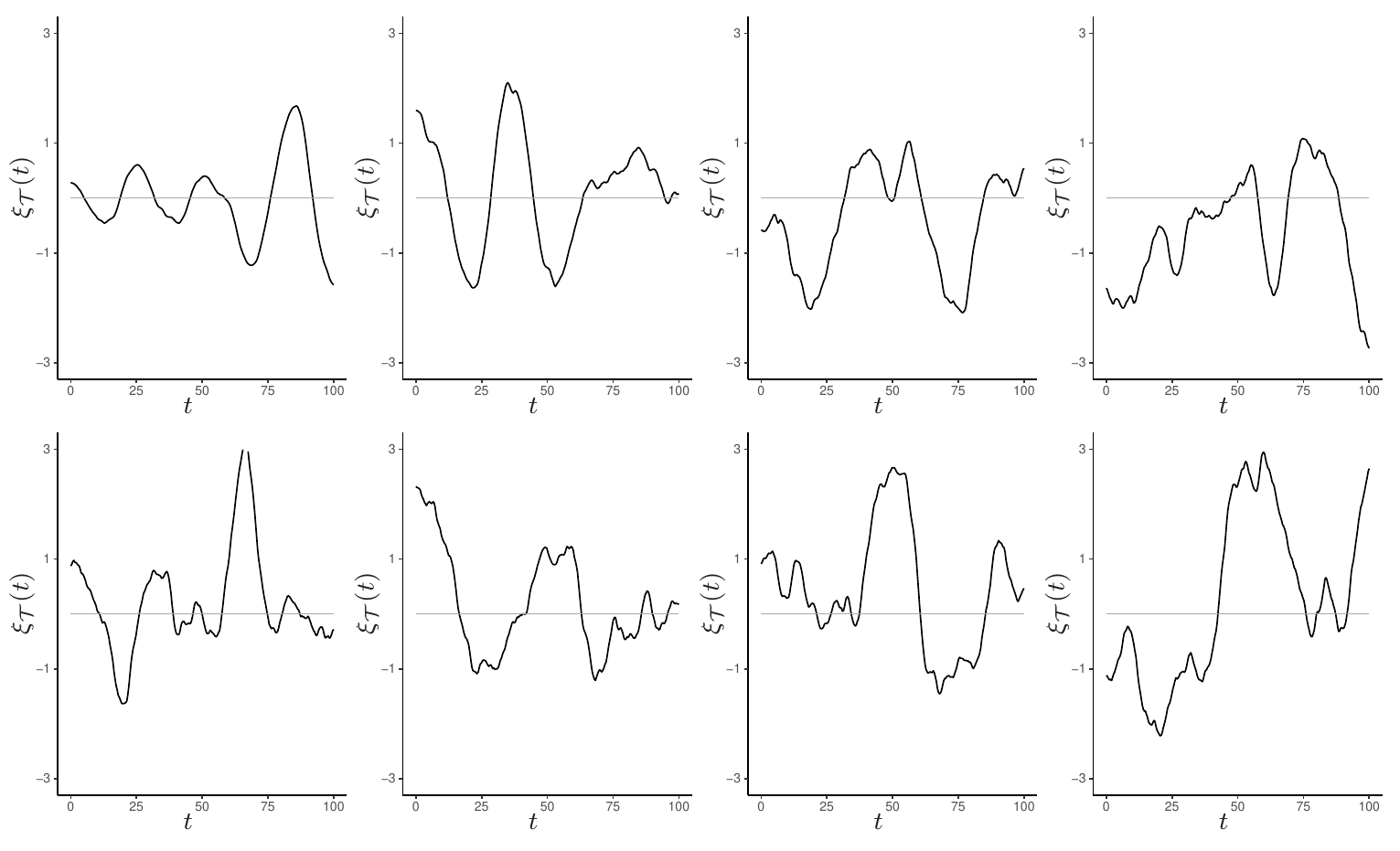}
  \caption{Realizations of $\xi_\TTT(t)$ where $\xi_\TTT$ is a
    zero-mean Gaussian random field with spectral density
    \eqref{eq:ps_T}, where $\alpha_\TTT=2$, $\kappa_\TTT = 2/10$ and
    $s_{\TTT}^2=1$. \textit{Top, left to right}:
    shown for $\varphi_\TTT=-9/10, -7/10, -4/10,
    -1/10$. \textit{Bottom, right to left}: shown for
    $\varphi_\TTT=1/10, 4/10, 7/10, 9/10$.
  }
  \label{fig:circ_fields}
\end{figure}
  \newpage
\begin{proposition}\label{prop:variance-on-R}
  A stationary Gaussian process $\xi_\TTT(t)$, $t\in\RR_+$, with power
  spectrum \eqref{eq:ps_T} has marginal variance
  $s_\TTT^2:=\text{\normalfont var}\{\xi_\TTT(t)\}$ given by 
  \[
    s_\TTT^2 =
    \frac{\sigma_\TTT^2}{4\kappa_\TTT^3}\left(\frac{2}{1+\vphi_\TTT}\right)^{1/2},\qquad
    \vphi_\TTT > -1,
  \]
  and is underdamped for $\vphi_\TTT\in(-1,1)$, critically damped for
  $\vphi_\TTT=1$, and overdamped for $\vphi_\TTT>1$.
\end{proposition}
Proofs of Propositions \ref{prop:variance-on-R2}
\ref{prop:variance-on-S1} and \ref{prop:variance-on-R} are given in
Section \ref{sec:variances} of the supplement.
\section{Practical computation and parameter construction}
\label{sec:computation}
\subsection{Numerical evaluation of log-likelihood functions}
\label{sec:integrals}
The logarithm of the likelihood function \eqref{eq:likelihood}
consists of two terms: a stochastic integral and the evaluation of the
random field at the data points. Hence, the log-likelihood is
analytically intractable as it requires the integral of the intensity
function which cannot be calculated explicitly.\ The integral can,
however, be approximated numerically using the approach of
\cite{simpsetal16} who introduced a computationally efficient method
for performing inference on log-Gaussian Cox processes based on
continuously specified finite-dimensional Gaussian random fields.\ The
remainder of this section describes the integration schemes that we
use for the integrals in the log-likelihood functions of models \MO,
\MOT, \MA, and \MB, that is, the integration scheme that is used for
the integrals
\begin{IEEEeqnarray*}{rCl}
  I_{\Omega} &=& -e^\beta\sum_{i=1}^{{N}}
  \int_{{t}_{i-1}}^{{t}_{i}}\exp\left[ \bm \psi^{\Omega
    }\{s(t)\}^\top\, \bm w^{\Omega}\right] \GOm(\dd t),\\
  I_{\Omega,\TTT} &=& -e^\beta\sum_{i=1}^{{N}}
  \int_{{t}_{i-1}}^{{t}_{i}}\exp\left[ \bm \psi^{\Omega
    }\{s(t)\}^\top\, \bm w^{\Omega
    } + \bm \psi^{\TTT}(t)^\top \bm w^\TTT\right] \GOm(\dd t),\\
  I_{\Omega\times\Theta} &=&
  -e^\beta\sum_{i=1}^{{N}}\int_{0}^T\exp\left[\bm \psi^{\Omega \times
      \Theta}\{s(t), \theta(t)\}^\top\, \bm w^{\Omega \times
      \Theta}\right]\GOm(\dd t), \quad\text{and}\\\
  I_{\Omega\times\Theta,\TTT} &=&
  -e^\beta\sum_{i=1}^{{N}}\int_{0}^T\exp\left[\sum \bm \psi^{\Omega
      \times \Theta}\{s(t), \theta(t)\}^\top\, \bm w^{\Omega \times
      \Theta} + \bm \psi^{\TTT}(t)^\top \bm w^\TTT\right] \GOm(\dd t).
\end{IEEEeqnarray*}

To ensure a numerically stable and computationally efficient
integration scheme, we use the construction of \cite{yuanetal17} and
split every straight line segment
$L_i = \{(s(t),\theta(t))\,:\,t\in[t_{i-1}, t_i)\} \subset
\Omega\times\Theta$ into straight line segments
$L_{ij}\subset \Omega\times\Theta$, $j=1,\dots, J_i$, $J_i \in \NN$,
each of which resides in exactly one prism in $\Omega\times\Theta$
that is shaped by a distinct triangle in the spatial mesh and a
distinct arc in the circular mesh.\ To simplify the notation, we set
$(\widetilde{L}_{i}\,:\,i=0,\dots,\widetilde{N})= (L_{ij}\,:\,i =
0,\dots,N, j=1,\dots,J_i)$, where the elements in the latter vector
appear in lexicographic order, so that $\widetilde{L}_{i-1}$ and
$\widetilde{L}_i$ correspond to adjacent line segments, and
\smash{$\widetilde{N}=\sum_{i=0}^NJ_{i}$}.\ This segmentation induces,
for $i=0,\dots, \widetilde{N}$, a set of times $\widetilde{t}_i$,
locations $s(\widetilde{t}_i)$ and head-direction angles
$\theta(\widetilde{t}_i)$, which are computed via linear
interpolation.\

Using the trapezoidal rule with based on the partition
$\{\widetilde{t}_{i}\}$ of the interval $\TTT$, yields 
\begin{IEEEeqnarray}{rCl}
  \IEEEeqnarraymulticol{3}{l} {e^{-\beta}\,I_\Omega\doteq
    -\sum_{i=1}^{\widetilde{N}} \frac{\lvert \widetilde{L}_i
      \rvert}{2}\left[\exp\left( \bm
        \psi^{\Omega}\{s(\widetilde{t}_{i-1})\}^\top\, \bm
        w^{\Omega}\right) + \exp\left( \bm
        \psi^{\Omega}\{s(\widetilde{t}_i)\}^\top\,\bm
        w^{\Omega}\right)\right]}
  \nonumber\\
  \qquad &\doteq - \sum_{i=1}^{\widetilde{N}} \frac{\lvert
    \widetilde{L}_i \rvert}{2}\left[\bm \psi^{\Omega}
    \{s(\widetilde{t}_{i-1})\}^\top \, \exp(\bm w^{\Omega}) + \bm
    \psi^{\Omega} \{s(\widetilde{t}_{i})\}^\top \, \exp(\bm
    w^{\Omega}) \right] = -(\bm b^\Omega)^\top \exp(\bm w^\Omega),
  \label{eq:approx_exp}
\end{IEEEeqnarray}
where $\bm b_{\Omega}=(b_k^\Omega)\in \RR^{p_{\Omega}}$ with
$b_k^{\Omega} = \sum_{i=1}^{\widetilde{N}}
(\lvert\widetilde{L}_i\rvert/2)\sum_{j=i-1}^i
\psi_k^{\Omega}\{s(\widetilde{t}_j)\}$.\ We remark that the last
approximation in expression \eqref{eq:approx_exp} follows after
replacing
\smash{$h_j(t)=\exp\left[ \bm
    \psi^{\Omega}\{s(\widetilde{t}_{j})\}^\top\, \bm
    w^{\Omega}\right]$} by its linear interpolant which agrees with
$h_j(t)$, $j\in\{i-1,i\}$, at the corners of the triangle within which
the projection of $\widetilde{L}_i$ on $\Omega$ resides. The
integration scheme that is used for $I_{\Omega \times \Theta}$ is
identical to that used for $I_{\Omega}$ though it is useful to note
that the aforementioned approximation based on linear interpolation is
done in the same way except now the linear interpolant matches the
function at the corners of the prism in $\Omega\times \Theta$ within
which $\widetilde{L}_i$ resides.\ This gives the approximation
$e^{-\beta}\,I_{\Omega\times \Theta} \doteq -(\bm
b^{\Omega\times\Theta})^\top \exp(\bm w^{\Omega\times\Theta})$ where
$\bm b^{\Omega\times\Theta}\in\RR^{p_{\Omega}\, p_{\Theta}}$ with
$b_k^{\Omega\times\Theta} = \sum_{i=1}^{\widetilde{N}}
(\lvert\widetilde{L}_i\rvert/2)\sum_{j=i-1}^i
\psi_k^{\Omega\times\Theta}\{s(\widetilde{t}_j),\theta(\widetilde{t}_j)\}$.

The creation of the sequence of successive time points
\smash{$\{\widetilde{t}_i\}_{i=0}^{\widetilde{N}}$} allows a simple
construction of a mesh on $\TTT$ that can be used to construct
$h_\TTT=\log \xi_\TTT$ in models $\MOT$ and $\MB$.\ Although
$\{\widetilde{t}_j\}_{i=0}^{\widetilde{N}}$ may be used directly as a
temporal mesh for the construction of $\xi_\TTT$, there is a
computational caveat with such an approach since the requirement of
having a fine resolution both for the spatial and directional meshes,
which is needed to ensure the GMRF approximation is accurate, may lead
to an extremely large number of knots on $\TTT$. Therefore, the set
$\{\widetilde{t}_i\}$ is thinned by choosing a subsequence
\smash{$\{\widetilde{t}_{i_j}\}_{j=0}^{p_{\TTT}-1}$}.\ This allows to
reduce the number of mesh knots on $\TTT$ and ensures that each line
segment
$\{(t,s(t))\,:\, t\in[\widetilde{t}_{i-1}, \widetilde{t}_{i})\}
\subset \Omega\times\TTT$ resides in one, and only one, prism shaped
by a distinct triangle in the spatial mesh, a distinct arc in the
circular mesh and a distinct time interval in the temporal mesh.\
Using this temporal mesh, then the composite trapezoidal rule with
partition \smash{$\{\widetilde{t}_{i}\}_{i=0}^{\widetilde{N}}$} yields
the approximation
\begin{IEEEeqnarray*}{rCl}
  e^{-\beta}\,I_{\Omega, \TTT}&\doteq& -\exp(\bm w^\TTT)^\top \bm
  B^{\Omega,\TTT} \exp(\bm w^{\Omega}), \qquad\text{where }\bm
  B^{\Omega,\TTT} = (b_{rc}^{\Omega,\TTT}) \in \RR^{p_{\TTT}\times
    p_{\Omega}}, \text{ with}\\
  \quad b_{rc}^{\Omega , \TTT} &=& \sum_{i=1}^{\widetilde{N}}
  (\lvert\widetilde{L}_i\rvert/2)\sum_{j=i-1}^i
  [\psi_c^{\Omega}\{s(\widetilde{t}_j)\}\,
  \psi_r^{\TTT}(\widetilde{t}_j)], \qquad r=1,\dots, p_{\TTT}, \quad
  c=1,\dots, p_{\Omega}
\end{IEEEeqnarray*}
These approximations yield the following approximate likelihoods
\begin{IEEEeqnarray*}{rCl}
   \ell(\bm \theta^{\Omega}\mid\MO ) &\doteq& -e^\beta\, (\bm
  b^{\Omega})^\top \exp(\bm w^{\Omega}) + n \beta + \bm 1_n^\top \bm
  A^{\Omega}_{\text{obs}} \, \bm
  w^{\Omega}\\
   \ell(\bm \theta^{\Omega,\TTT}\mid\MOT ) &\doteq& -e^\beta\,
  \exp(\bm w^{\TTT}) \bm B^{\Omega,\TTT} \exp(\bm w^{\Omega}) + n
  \beta + \bm 1_n^\top \bm A^{\Omega,\TTT}_{\text{obs}} \, \bm
  w^{\Omega,\TTT}\\
   \ell(\bm \theta^{\Omega\times \Theta} \mid \MA) &\doteq&
  -e^\beta\, (\bm b^{\Omega\times \Theta})^\top \exp(\bm
  w^{\Omega \times \Theta}) + n \beta + \bm 1_n^\top \bm
  A^{\Omega\times \Theta}_{\text{obs}} \, \bm w^{\Omega \times
    \Theta} \\
   \ell(\bm \theta^{\Omega\times \Theta, \TTT}\mid \MB) &\doteq&
  -e^\beta\, \exp(\bm w^\TTT)^\top \bm
  B^{\Omega\times\Theta,\TTT} \exp(\bm w^{\Omega \times \Theta})+ n
  \beta + \bm 1_n^\top (\bm A_{\text{obs}}^\TTT \, \bm w^\TTT+ \bm
  A^{\Omega\times \Theta}_{\text{obs}} \, \bm w^{\Omega \times
    \Theta})
\end{IEEEeqnarray*}
for $\MO, \MOT, \MA$ and $\MB$, respectively.\ We remark that the
resulting approximate posterior distributions are close to the true
posterior since all expressions for the likelihoods can be written in
the form given by equation (3) in \citet{simpsetal16}.

\subsection{Prior distributions of hyperparameters}
\label{sec:prior}
The full Bayesian hierarchical model is completed upon specification
of the joint prior distribution of all hyperparameters in the model.\
We consider the transformations
$\rho_{\Omega} = \sqrt{8}/\kappa_\Omega, \rho_{\Theta} = \sqrt{8 \,
  (3/2)}/\kappa_\Theta$ and
$\rho_{\TTT} = \sqrt{8 \, (3/2)}/\kappa_\TTT$.\ The choice of
$\sqrt{8}$ and $\sqrt{8(3/2)}$ in the definition of $\rho_\Omega$ and
$\rho_\TTT$ follows \cite{lindetal11} and makes $\rho_\Omega$ and
$\rho_\TTT$ the distance and time at which the correlation of the
processes $\xi_\Omega$ and $\xi_\TTT$ is approximately 0.1,
respectively.\ For the cyclic processes $\xi_\Theta$, the
interpretation of $\kappa_\Theta$ depends on whether this parameter is
large or small.\ For large $\kappa_\Theta$, the cyclic aspect can be
ignored, and the interpretation in the cyclic Whittle-Mat\'{e}rn model
is that the correlation range is also
$\sqrt{8\nu_\Theta} /\kappa_\Theta$, where
$\nu_\Theta=\alpha_\Theta - d/2 = 3/2$ denotes the smoothness
parameter of the prior process.\ 
If $\sqrt{8\nu_\Theta} /\kappa_\Theta$ is small relative to $2\pi$,
then this value is proportional to the correlation range, that is, it
is again the distance at which the correlation is approximately 0.1,
and this also holds for a wide range of $\nu_\Theta$-values and not
only for the special case $\nu_\Theta=3/2$ that we consider here.\
However, when $\kappa_\Theta$ is small enough for that nominal
correlation range to become close to or larger than $2\pi$, the cyclic
behaviour of the process kicks in and the correlation no longer gets
close to zero, instead the covariance function $\zeta_\Theta$ grows
close to an increasing constant, plus a cosine function and additional
terms associated with higher frequencies in the spectrum.\ This latter
case did not appear to be particularly troublesome in our experiments
and hence, we proceed by using the aforementioned reparameterizations.

In what follows, a random variable $\phi$ is said to follow a
location-scale-Beta$(a, b, \ell, u)$ distribution, if
$\phi=^d \ell + (u-\ell) X$, where $X\sim \text{Beta}(a, b)$, with
$a,b > 0$. We assume
\begin{IEEEeqnarray*}{rClrrClrrCl}
  \varphi_{\Omega} &\sim& \text{location-scale-Beta}(a_{\Omega},
  b_{\Omega}, \ell, u),&\qquad&\varphi_{\Theta}&\sim&
  \delta_{\{1\}} ,&\qquad& \varphi_\TTT&\sim&\delta_{\{1\}}\\
  \rho_\Omega &\sim& \text{log-Normal}(\mu_{\Omega},
  \varsigma_{\Omega}^2),&\qquad& \rho_\Theta&\sim&
  \text{Exp}(\eta_\Theta), &\qquad& \rho_\TTT&\sim&\text{Exp}(\eta_\TTT),\\
  s_\Omega &\sim& \text{Exp}(\nu_\Omega),&\qquad& s_\Theta &\sim&
  \text{Exp}(\nu_\Theta), &\qquad& s_\TTT&\sim&\text{Exp}(\nu_\TTT),
\end{IEEEeqnarray*}
where $\delta_{\{1\}}$ denotes the Dirac delta distribution at
$\{1\}$, $\ell=-1$, $u=1$, $a_\Omega$, $b_\Omega$, $\varsigma_\Omega$,
$\eta_\Theta$, $\eta_\TTT$, $\nu_\Omega$, $\nu_\Theta$,
$\nu_\TTT \in \RR_+$ and $\mu_\Omega\in\RR$.\ The rationale for
choosing an exponential distribution for the standard deviation
parameters $s_\Omega$, $s_\Theta$ and $s_\TTT$, and the correlation
range parameters $\rho_\Theta$ and $\rho_\TTT$, is the the memoryless
property, indicating that we are, a priori, claiming relative
ignorance about the marginal standard deviation and the scale of the
process.\ On the other hand, a log-Normal prior for $\rho_\Omega$ was
chosen to include prior knowledge on the range of the spatial process,
by centering the distribution to sensible range values.


Regarding the damping coefficients $\varphi_{\Omega}$,
$\varphi_\Theta$ and $\varphi_\TTT$, we initially used a
location-scale-Beta with $\ell=-1$ and $u=1$ so that the prior
distribution is supported on the interval $(-1,1)$.\ That is, we
initially constrained all damping coefficients so that the prior
Gaussian random fields associated with the spatial, directional and
temporal effects were underdamped.\ Our results confirmed strong
evidence for underdamped behaviour in the spatial effect, but no
evidence for underdamped behaviour in the directional and temporal
effects. Hence, in what follows we use a degenerate prior distribution
at $\{1\}$.


\if0\blind{
We discuss a general way of constructing process models here by via
precision operators.\ This is directly connected to some of Wahbas
generalised spline penalty constructions that directly define spline
penalties via Sobolev space norms that are induced by inner products.

Let $\langle\, \cdot \; , \; \cdot\, \rangle$ denote the function
inner product for vector
fields 
$\int u(s)^\top v(s) ds$.\ Then for $(\kappa^2 - \Delta) x = \noise$
with Neumann boundaries, the precision operator is
\begin{equation}
<  Q(u,v) := \langle\kappa^2 u - \Delta u, \kappa^2 v - \Delta v\rangle
  = \kappa^4 \langle u,v\rangle+ 2\kappa^2\langle\nabla u, \nabla
  v\rangle + \langle-\Delta u, -\Delta v\rangle
\end{equation}
where the middle gradient-gradient inner product comes from Stokes'
theorem.\ For the oscillating processes we construct,
$Q(u,v) = \langle\kappa^2 u + \Delta u, \kappa^2 v + \Delta v\rangle +
2(1+\varphi)\kappa^2\langle\nabla u, \nabla v\rangle$ which coincides
with the above for $\varphi=1$, as expected.\ The first term is the
undampened oscillation precision, and the second term is the dampening
term.

If we let $U(\omega)$ and $V(\omega)$, $\omega \in \RR^d$, denote the
spectral versions of $u$ and $v$, we get (with
$\langle U,V\rangle=\int \overline{U(\omega)} V(\omega) d \omega$)
$Q(u,v) = \langle\kappa^2 U - |\omega|^2 U, \kappa^2 V - |\omega|^2
V\rangle + 2(1+\varphi)\kappa^2\langle i \omega U, i \omega
V\rangle$. This reveals a sense in which the $i \omega$ transfer
function appears.\ Similarly, we may reformulate a phase-shifted
version as
$\langle i \omega U, i \omega V\rangle = \int \overline{i \omega U} i
\omega V d\omega = \int \omega^t \omega \overline{U} \overline{V}
d\omega = \int |\omega|^2 \overline{U} V d\omega = \langle |\omega| U,
|\omega| V \rangle$ or $\langle i|\omega| U, i|\omega| V\rangle$,
which shows a connection of the dampening term to both $|\omega|$ and
$i|\omega|$. Specifically, the $|\omega|=\kappa$ frequencies that are
in the null space of the first term of the precision are dampened by
$2(1+\varphi)\kappa^4$ by the second term.\ We know that $|\omega|$ is
the transfer function for $(-\Delta)^{1/2}$, so that's the spatial
operator of the dampening.\ But unlike the 1-dimensional case, in the
general dimension case this doesn't give us a simple $L u = \noise$
form for the SPDE as a whole, but instead just a precision operator
with an undampened oscillation, plus a dampening term, where the
latter can be written either
as$\langle (-\Delta)^{1/2} u, (-\Delta)^{1/2} v\rangle$ or as
$\langle \nabla u, \nabla v\rangle$.\ \color{black}

$f(x) \partial_n g(x) |_{x=a} + f(x) \partial_n g(x) |_{x=b} = 0$
since $f(a)=f(b)$ and $\partial_n g(a) = -\partial_n g(b)$, since
$g'(a)=g'(b)$}
\fi

\section{Case study}
\label{sec:cs}
\subsection{The grid-cell data}
\label{sec:data}
We used data from a study that reports the activity of grid cells
in the medial entorhinal cortex of mice exploring a square arena with a polarizing cue on one of the
walls (\cite{gerlei2020grid}. \ The recordings were made with tetrodes inserted into the brain and attached to
chronically implanted microdrives.\ Recording sessions were
for a maximum of 90 minutes.\ A camera was attached above the arena to track the position and head-direction of mice. \ The voltage data were aquired using the OpenEphys platform \citep{siegle2017open}, spike events detected and then clustered using MountainSort
\citep{chung2017fully}.\ The detected clusters were evaluated based on quality
metrics for isolation, noise-overlap and peak signal to noise ratio
\citep{chung2017fully}. Clusters with firing rate $> 0.5\text {Hz}$,  isolation
$> 0.9$ and noise overlap $<0.05$ were retained.\

Grid cells were identified using established metrics
\citep{sargolini2006conjunctive}. Pure grid cells were identified
as neurons with a grid score $\geq 0.4$ and head-direction cells as
cells with a head-direction score $\geq 0.5$.\ Conjunctive grid cells
were defined as cells that passed both the grid cell and
head-direction cell criteria. Here, we report data from 1 pure grid
cell from 1 mouse.

\subsection{Cross-validation}
\label{sec:cv}
We compare the predictive performance of each model via 50-50\%
cross-validation.\ The full trajectory of the mouse is split up in to
connected segments of fixed lengths, half of which are selected as
training data.\ Specifically, we partition the time domain
$\TTT=[0,T]$ into $M$ consecutive non-overlapping intervals
$A_1,\dots,A_M$ of size $\tau=T/M$ seconds.\ The two training sets are
defined via the observed point pattern in $A_1, A_3,\dots$ and in
$A_2,A_4,\dots$, respectively, whilst the test sets are formed by the
complementary intervals, that is, by $A_2,A_4,\dots$ and by
$A_1, A_3,\dots$, respectively.\ 
\begin{figure}[htpb!]
  \centering
  \includegraphics[scale=1]{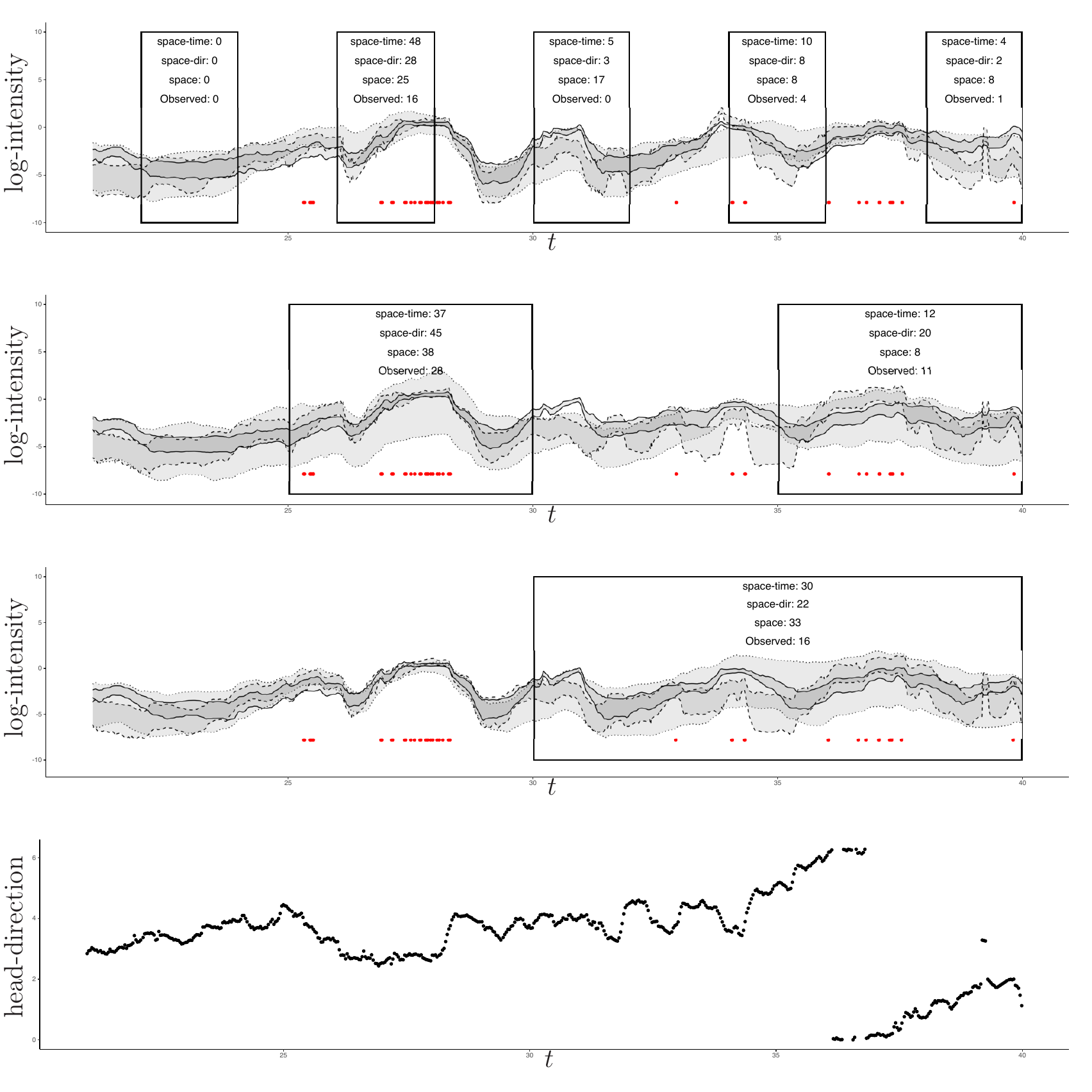}
  \caption{Illustration of output from cross-validation for models
    $\MO, \MA$ and $\MOT$, for the time period 20--40 seconds, and for
    three time split sizes, first row: $\tau=2$s, second row:
    $\tau=5$s, third row: $\tau=10$s.\ The last row shows the
    head-direction covariate over the same time period.  Solid-lines,
    dashed-lines and dotted lines show point-wise 95\% credible
    intervals of the logarithm of the intensity function
    $\lambda_\TTT(t)$ from models $\MO, \MA$ and $\MOT$, respectively.
    The test set is the union of time intervals associated with the
    boxed regions. The training set is the complement of the test
    set.  The statistics in the boxed regions show the mean of the
    posterior predictive distribution of the number of firings on the
    segment that each box corresponds to.\ The observed point pattern
    is superimposed on the first three panels and is shown with red dots. }
  \label{fig:cv_log_intensity}
\end{figure}

We use proper scoring rules \citep{gneit07} to evaluate the
performance of models based on their ability to predict the number of
firing events on segments from the training
set.\ 
The scoring rules that we
consider are negatively oriented so that smaller scores correspond to
better predictions.\ A negatively oriented scoring rule $S$ is called
proper if $\mathbb{E}\{S(G, Y)\} \leq \mathbb{E}\{S(F, Y)\}$ when
$Y\sim G$, i.e., on average, the true distribution $G$ will not give a
worse score than any other distribution $F$. 
Both of the scores that we consider, namely the squared error and
Dawid--Sebastiani scores, are proper.
\begin{table}[htbp!]
  \centering
  \begin{tabular}{r||lll|lll}
    \hline 
    &\multicolumn{3}{c|} {SE}& \multicolumn{3}{c} {DS}\\
    \hline\hline
    split & \MA & \MOT & \MB  &\MA  & \MOT & \MB  \\ 
    \hline\hline
    \multirow{6}{*}{Fold 1$\quad$}
   2s   & $-$24 & 76 & 139 & $-$6 & $-$8 & $-$8 \\              
   5s   & $-$65 & 1027 & 726 & $-$9 & $-$13 & $-$13 \\          
   10s  & $-$175 & 2767 & 2687 & $-$15 & $-$19 & $-$20 \\       
   20s  & $-$404 & 4643 & 2169 & $-$8 & $-$14 & $-$15 \\        
   30s  & $-$1520 & 1261 & $-$862 & $-$15 & $-$26 & $-$27 \\      
   40s  & $-$1281 & 4183 & 272 & $-$7 & $-$8 & $-$9 \\
        \hline
        \multirow{6}{*}{Fold 2$\quad$} 
   2s   & $-$29 & 129 & 81 & $-$8 & $-$11 & $-$9 \\             
   5s   & $-$41 & 307 & 280 & $-$22 & $-$28 & $-$29 \\          
   10s  & $-$105 & 503 & 322 & $-$14 & $-$17 & $-$17 \\         
   20s  & $-$727 & 2452 & 2208 & $-$29 & $-$103 & $-$99 \\      
   30s  & $-$561 & 7020 & 1548 & $-$9 & $-$10 & $-$12 \\        
   40s  & $-$1817 & 9516 & 724 & $-$21 & $-$27 & $-$28 \\
        \hline
        \multirow{6}{*}{Average$\quad$}
   2s   & $-$26 & 102 & 110 & $-$7 & $-$9 & $-$8 \\             
   5s   & $-$53 & 667 & 503 & $-$16 & $-$21 & $-$21 \\          
   10s  & $-$140 & 1635 & 1505 & $-$14 & $-$18 & $-$18 \\       
   20s  & $-$566 & 3548 & 2189 & $-$19 & $-$59 & $-$57 \\       
   30s  & $-$1041 & 4140 & 343 & $-$12 & $-$18 & $-$19 \\       
   40s  & $-$1549 & 6849 & 498 & $-$14 & $-$18 & $-$19 \\       
  \end{tabular}
  \caption{Average squared-error (SE) and Dawid--Sebastiani (DS) score
    difference between $\mathcal{M}$ and $\MO$ for
    $\mathcal{M}$$\in\{\MA,
    \MOT,\MB\}$, rounded to the nearest integer.}
    \label{tab:score-diff}
\end{table}
Suppose that
$n_i$ is the observed number of firing events on segment
$i$ in the test set, with
$F_i^j$ the predictive distribution function for the number of
firings, $N_i^j$, on this segment for model $j$, $1 \leq i \leq M, j
\in \{\MO, \MOT, \MA, \MB\},$ where
$M$ is the total number of segments in the test set.  The number of
firings, $N_i^j$, on segment $i$ for model
$j$, conditional on the latent parameters, is Poisson distributed with
mean $\int_{\mathcal{T}} \lambda_{\TTT}(t) ~\dd
t$, where
$\mathcal{T}$ is a time segment in the test set.  The predictive
probability mass function, $f_i^j$, for
$N_i^j$ is obtained by integrating out the latent parameters from the
product of the Poisson distribution and the posterior of the latent
parameters, which is an intractable integral.\ Although this Poisson
mixture distribution for
$N_i^j$ is not available in closed form, its mean and variance may be
obtained from the standard identities
\begin{align} \mathbb{E}(N_i^j) &= \mathbb{E}\{\mathbb{E}(N_i^j \mid
  \bm{x}^j)\} \label{ModExp} \\ \text{var}(N_i^j) &=
  \mathbb{E}\{\text{var}(N_i^j \mid \bm{x}^j)\} +
  \text{var}\{\mathbb{E}(N_i^j \mid \bm{x}^j)\} \label{ModVar}
\end{align} where $\bm{x}^j$ is the vector of latent parameters for
model $j$, and the expectations are all conditional on the training set.\
We may obtain high precision estimates of (\ref{ModExp})
and (\ref{ModVar}) via Monte Carlo simulation. In particular, if
$\bm{x}^j_1,\ldots,\bm{x}^j_K$ are $K$ independent realizations of the
vector of latent parameters for model $j,$ then we estimate
$\mathbb{E}(N_i^j)$ using
\begin{equation} \widehat{\mathbb{E}}(N_i^j) =
\frac{1}{K}\sum_{k=1}^{K}\mathbb{E}(N_i^j \mid \bm{x}^j_k)
\label{eq:posterior-expectation}
\end{equation}
where $\mathbb{E}(N_i^j \mid
\bm{x}^j_k)$ are obtained using the same numerical integration method
as used in Section~\ref{sec:integrals}. We estimate
$\text{var}(N_i^j)$ using
\begin{equation} \widehat{\text{var}}(N_i^j) =
  \frac{1}{K}\sum_{k=1}^{K}\mathbb{E}(N_i^j \mid \bm{x}^j_k) +
  \frac{1}{K}\sum_{k=1}^{K}\big\{\mathbb{E}(N_i^j \mid \bm{x}^j_k) -
  \widehat{\mathbb{E}}(N_i^j) \big\}^2
\label{eq:posterior-variance}
\end{equation} where in the first sum we have used the fact that
the conditional mean and variance of $N_i^j$ given $\bm{x}^j_k$ are
equal since $N_i^j$ given $\bm{x}^j_k$ is Poisson distributed.\ In the
\texttt{inlabru} package \citep{inlabru}, replications of the latent
field can be obtained using the \texttt{generate} function.

We consider scoring rules that only depend on the mean and variance of
the predictive distribution, as these are the only quantities that we
easily have access to.\ In particular, we consider the squared error
and Dawid--Sebastiani \citep{dawid1999coherent} scores, $S_{\text{SE}}$
and $S_{\text{DS}}$ respectively, defined by
\begin{align} S_{\text{SE}}(F_i^j, n_i) &= (n_i -
  \mu_i^j)^2 \label{SE_score} \\ S_{\text{DS}}(F_i^j, n_i) &=
  \bigg(\frac{n_i - \mu_i^j}{\sigma_i^j}\bigg)^2 +
  \log\{(\sigma_i^j)^2\} \label{DS_score}
\end{align} where $\mu_i^j = \mathbb{E}(N_i^j)$ and $\sigma_i^j =
\{\text{Var}(N_i^j)\}^{1/2}$ and $N_i^j \sim F_i^j$.\ As the exact
values of $\mu_i^j$ and $(\sigma_i^j)^2$ are intractable, we replace them with
their  estimates from expressions \eqref{eq:posterior-expectation} and \eqref{eq:posterior-variance}.

To compare the predictive performance of models
$\mathcal{M}\in\{\MA, \MOT, \MB\}$ with the reference model $\MO$,
using a given score $S$ that may be either $S_{\text{SE}}$ or
$S_{\text{DS}}$, we compute the set of model scores
$\{S(F_i^{\MO}, n_i,)\}_{i=1}^{M}$ and
$\{S(F_i^{\mathcal{M}}, n_i,)\}_{i=1}^{M}$ associated with models
$\MO$ and $\mathcal{M}$, respectively. The model with lowest mean
score is then preferred. To decide whether or not there is
statistically significant evidence in favour of one model over the
other we consider a hypothesis test with null hypothesis that the
model scores $S(F_i^{\mathcal{M}}, n_i)$ and $S(F_i^{\MO}, n_i)$ are
pairwise exchangeable for each $1 \leq i \leq M$, versus the
alternative that $\MO$ is on average worse than $\mathcal{M}$. The
test statistic $T_{\text{test}}$ used for comparing between models
$\MO$ and $\mathcal{M}$, is the observed mean difference in scores,
with negative values indicating that the $\mathcal{M}$ is better than
the $\MO$. Under the null hypothesis $S(F_i^{\MO}, n_i)$ and
$S(F_i^{\mathcal{M}}, n_i)$ are identically distributed, and so the
observed score $S(F_i^{\MO}, n_i)$ would be equally likely to be
produced by model $\mathcal{M},$ and similarly,
$S(F_i^{\mathcal{M}}, n_i)$ is equally likely to be produced by
$\MO$. Thus, under the null hypothesis, for any given score difference
$S^-_i = S(F_i^{\mathcal{M}}, n_i) - S(F_i^{\MO}, n_i)$, it would be
equally likely that $-S^-_i$ was observed. This observation is used to
construct a randomized testing procedure, given in Section
\ref{sec:perm-test} of the supplement, which allows us to approximate
the distribution of $T_{\text{test}}$ under the null hypothesis.\ The
output of the test is an unbiased estimate of the one sided $p$-value
for the test with null hypothesis that the scores $S(F_i^{\MO}, n_i)$
and $S(F_i^{\mathcal{M}}, n_i)$ are pairwise exchangeable for each
$1 \leq i \leq M$.\ Small values of $p$ give evidence against the null
hypothesis in favour of $\mathcal{M}$.\ Conversely, values of $p$
close to $1$ give evidence against the null hypothesis in favour of
$\MO$.\ In our applications of Algorithm \ref{TestAlg} we take
$J=10^6$.
\begin{table}[htbp!]
  \centering
  \begin{tabular}{r||lll|lll}
    \hline 
    &\multicolumn{3}{c|} {SE}& \multicolumn{3}{c} {DS}\\
    \hline\hline
    interval & \MA & \MOT & \MB  &\MA  & \MOT & \MB  \\ 
    \hline\hline
    \multirow{6}{*}{Fold 1$\quad$}
    2s  & 0 & 0.96 & 0.96 & 0 & 0 & 0 \\ 
    5s  & 0.02 & 1.00 & 0.92 & 0 & 0 & 0 \\ 
    10s & 0.07 & 1.00 & 1.00 & 0 & 0 & 0 \\ 
    20s & 0.05 & 1.00 & 0.89 & 0 & 0 & 0 \\ 
    30s & 0.05 & 0.83 & 0.11 & 0 & 0 & 0 \\ 
    40s & 0.01 & 1.00 & 0.63 & 0 & 0 & 0 \\
    \hline
    \multirow{6}{*}{Fold 2$\quad$}
    2s  & 0 & 0.98 & 0.71 & 0 & 0 & 0 \\
    5s  & 0.12 & 0.99 & 0.93 & 0 & 0 & 0 \\
    10s & 0.11 & 0.80 & 0.53 & 0 & 0 & 0 \\
    20s & 0.08 & 0.97 & 0.72 & 0 & 0 & 0 \\
    30s & 0.04 & 1.00 & 0.90 & 0 & 0.01 & 0 \\
    40s & 0.05 & 1.00 & 0.69 & 0 & 0 & 0 \\
    \hline
    \multirow{6}{*}{Folds combined$\quad$}
    2s  & 0 & 1.00 & 0.96 & 0 & 0 & 0 \\
    5s  & 0.01 & 1.00 & 0.98 & 0 & 0 & 0 \\
    10s & 0.02 & 1.00 & 0.99 & 0 & 0 & 0 \\
    20s & 0.01 & 1.00 & 0.91 & 0 & 0 & 0 \\
    30s & 0.01 & 1.00 & 0.65 & 0 & 0 & 0 \\
    40s & 0 & 1.00 & 0.72 & 0 & 0 & 0 \\
      \hline
    \end{tabular}
    \caption{$p$-values of the hypothesis that the squared-error (DS) and
      Dawid--Sebastiani (DS) scores between model $\mathcal{M}$ and
      $\MO$, $\mathcal{M}\in\{\MA, \MOT,
      \MB\}$, are pairwise exchangeable.}
      \label{tab:score-p-values}
\end{table}

Figure~\ref{fig:cv_log_intensity} illustrates the design of the
cross-validation scheme for a selection of models and interval sizes
fitted to the spike-train data that we analyze in in
Section~\ref{sec:analysis}.\ Table~\ref{tab:score-diff} shows the
average score differences for the two cross-validation folds, and
their overall averages, for six interval sizes $\tau=2$s, $5$s, $10$s,
$20$s, $30$s, and $40$s.\ The average SE score differences are all
negative for $\mathcal{M}_{\Omega\times\Theta}$ and positive for
$\mathcal{M}_{\Omega,\mathcal{T}}$ and
$\mathcal{M}_{\Omega\times\Theta,\TTT}$, indicating that the
temporally modulated models are, on average, less accurate in their
predictions than $\MO$.\ However, for the DS score differences, the
values are negative for all three models, indicating that the
additional model complexity provides a better match for the overall
variability of the system, as the DS score takes into account both
accuracy (bias) and precision (variability) of the predictions.\ These
results are further confirmed by the $p$-values shown in
Table~\ref{tab:score-p-values}.\ The proportions of negative score
differences, shown in Section \ref{sec:table_and_figures} of the
supplement (Table~\ref{tab:score-negative-proportion}), illustrate how
often each model had a better score than the reference model $\MO$,
for SE and DS, respectively. 

\subsection{Analysis}
\label{sec:analysis}
Here, we analyze further the spike-train of $n=6263$ firing events
shown in Figure \ref{fig:lphpp}.\ We fit all models from
Section~\ref{sec:pp_models} to the entire data set, using the values
$\mu_\Omega = 20$, $\varsigma_\Omega=0.4$, $a_\Omega = 2$,
$b_\Omega = 20$, $\nu_\Omega =1/2$, $\nu_\Theta = 1$,
$\nu_\TTT = 1/3$, $\eta_\Theta=1/(2\pi)$, and $\eta_\TTT=1/100$, for
the hyperparameters of the priors in Section \ref{sec:prior}.

Figure~\ref{fig:xiomega} shows the posterior mean and pointwise 95\%
credible intervals of $\xi_{\Omega}(s)$, $s\in\Omega$, for the
baseline model $\MO$ and its temporally modulated version $\MOT$o.\
The shape and the characteristics of the spatial effect overall agree
between the two models, but there are key differences such as the
smoothness and the uncertainty in the estimates, with $\MOT$ showing
less uncertain and smoother estimates.\ This indicates that the
inclusion of a temporally modulating factor captures overdispersion.\
Figure~\ref{fig:xiomega} shows the posterior mean of
$\xi_{\Omega\times\Theta}(s, \theta)$ as a function of $s\in\Omega$
for a range of $\theta$ values from $\MA$ and $\MB$.\ Pointwise 95\%
credible intervals are shown in Figures~\ref{fig:xiomega_CI_MA_low},
\ref{fig:xiomega_CI_MB_low}, \ref{fig:xiomega_CI_MA_upp} and
\ref{fig:xiomega_CI_MB_upp} of the supplement.\ Overall, the point and
interval estimates agree between the two models, with the estimates
from $\MB$ being smoother and less uncertain, which indicates again
that the temporally modulated model $\MB$ captures additional residual
variability in the intensity of the spike train.\ Allowing for an
interaction effect between space and head-direction shows that there
is considerable variation in the locations of the peaks, and in the
sizes and shapes of the fields locally at the peaks.\ Although the
spatial effect is strongly oscillating, the posterior estimates do not
rule out the possibility of ridges between the peaks, as for example,
in the bottom left and right regions of $\Omega$ where there can be
considerable intensity between fields of high intensity for certain
$\theta$ values.\
\begin{figure}[htpb!]
  \centering
  \includegraphics[scale=1]{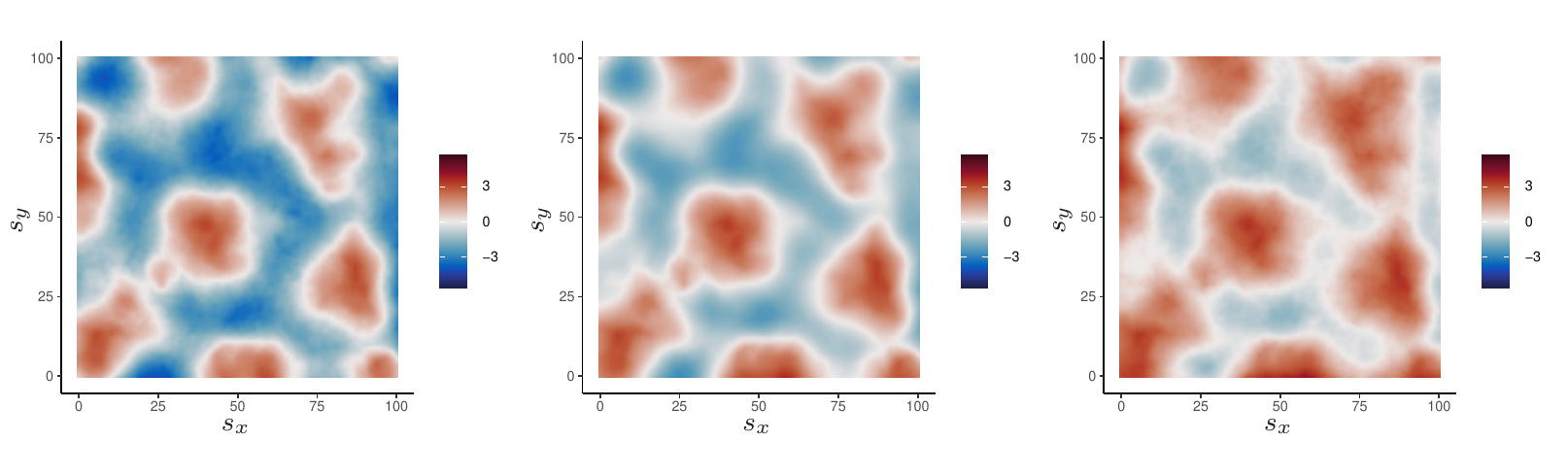}
  \caption{Posterior mean (center) and 95\% credible intervals (left:
    $0.025$ quantile, right: $0.975$) of $\xi_{\Omega}\{(s_x, s_y)\}$
    from $\MO$.}
  \includegraphics[scale=1]{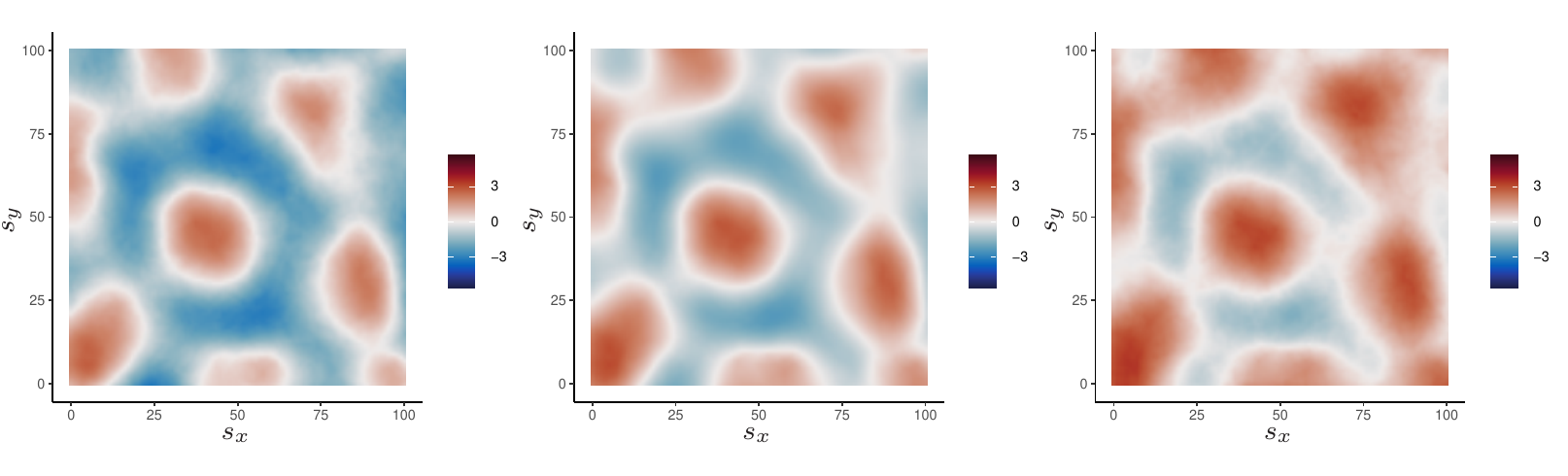}
  \caption{Posterior mean (center) and 95\% credible intervals (left: $0.025$ quantile, right: $0.975$) of
    $\xi_{\Omega}\{(s_x, s_y)\}$ from $\MOT$.} 
  \label{fig:xiomega}
\end{figure}
Last, Figure \ref{fig:xitheta} shows the posterior mean and 95\%
credible intervals of $\xi_{\Omega\times\Theta}(s, \theta)$ as a
function of $\theta\in\SSS$ for a range of $s$ values from $\MA$ and
$\MB$.\ The shape of the head-directional effect also varies with
location, and is more prominent in regions where the spatial effect is
low and less so in regions where the spatial effect is high.\ This is
also confirmed from Figure~\ref{fig:log_modulation_small} which shows
the posterior mean and 95\% credible intervals of the temporally
modulating component $\xi_\TTT(t)$ for a range of $t$ values from
$\MOT$ and $\MB$.\ The scatterplot shows that for the temporally
modulated models, the additional inclusion of a directional effect
shifts the posterior distribution of $\xi_\TTT$.\ When the modulating
effect $\xi_\TTT$ is low, then on average, the posterior mean from
$\MB$ is larger than the posterior mean from $\MA$ and vice versa.
\begin{figure}[htpb!]
  \centering
  \includegraphics[scale=1]{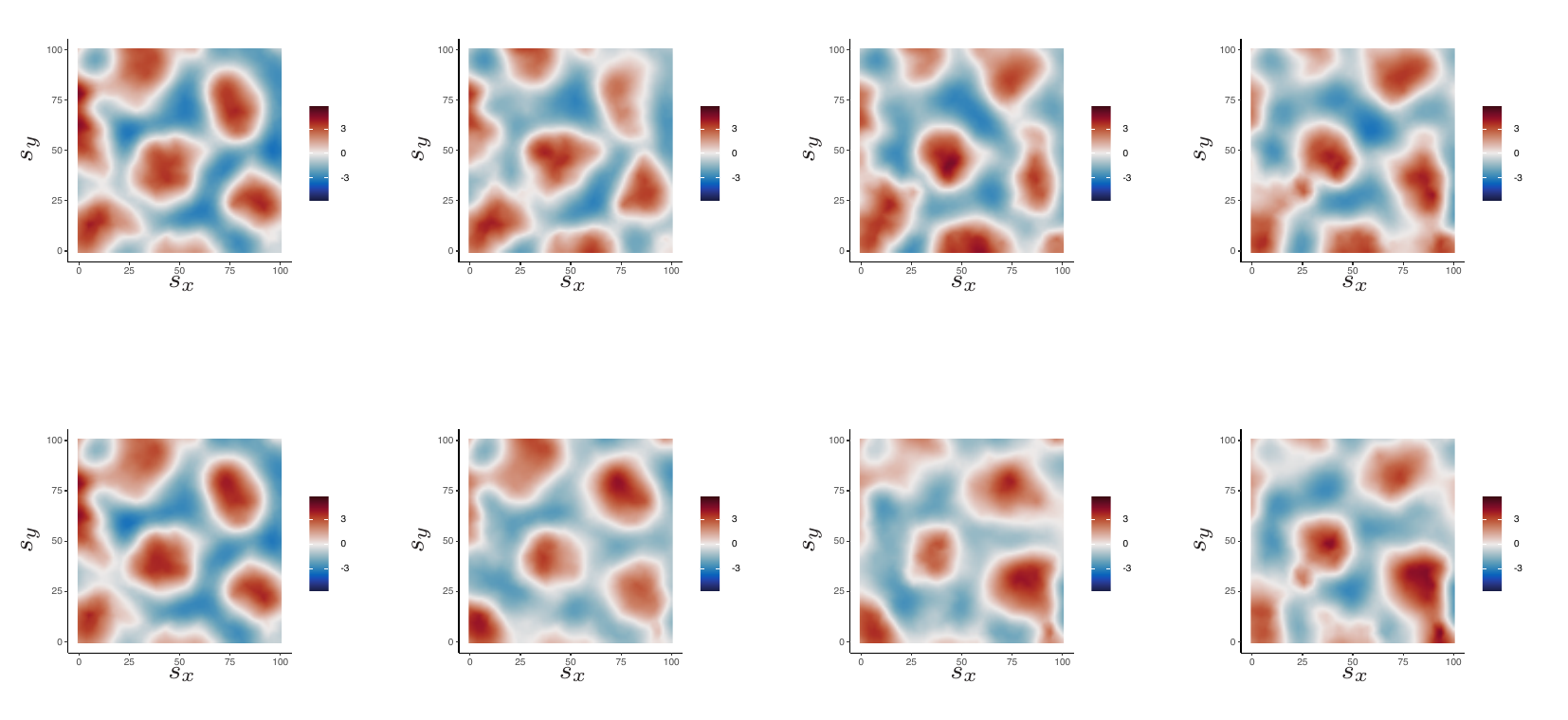}
  \caption{Posterior mean of
    $\xi_{\Omega\times\Theta}\{(s_x, s_y), \theta\}$ from $\MA$, shown
    in clockwise order for a sequence of equally spaced $\theta$
    values ranging from $\pi/8$ (top left) to $2\pi$ (bottom left).}
  \includegraphics[scale=1]{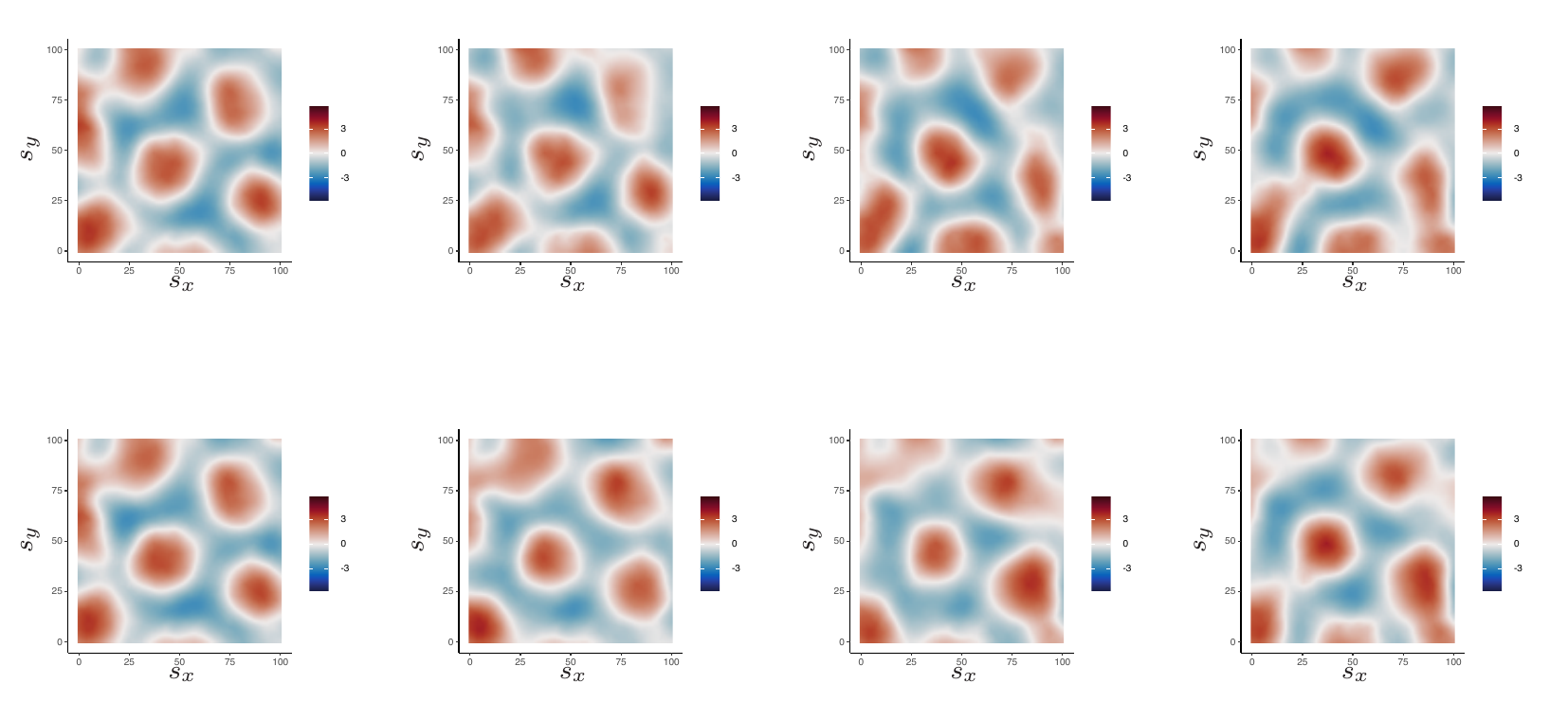}
  \caption{Posterior mean of
    $\xi_{\Omega\times\Theta}\{(s_x, s_y), \theta\}$ from $\MB$, shown
    in clockwise order for a sequence of equally spaced $\theta$
    values ranging from $\pi/8$ (top left) to $2\pi$ (bottom left).}
  \label{fig:xiomega}
\end{figure}

\begin{figure}[htpb!]
  \centering
  \includegraphics{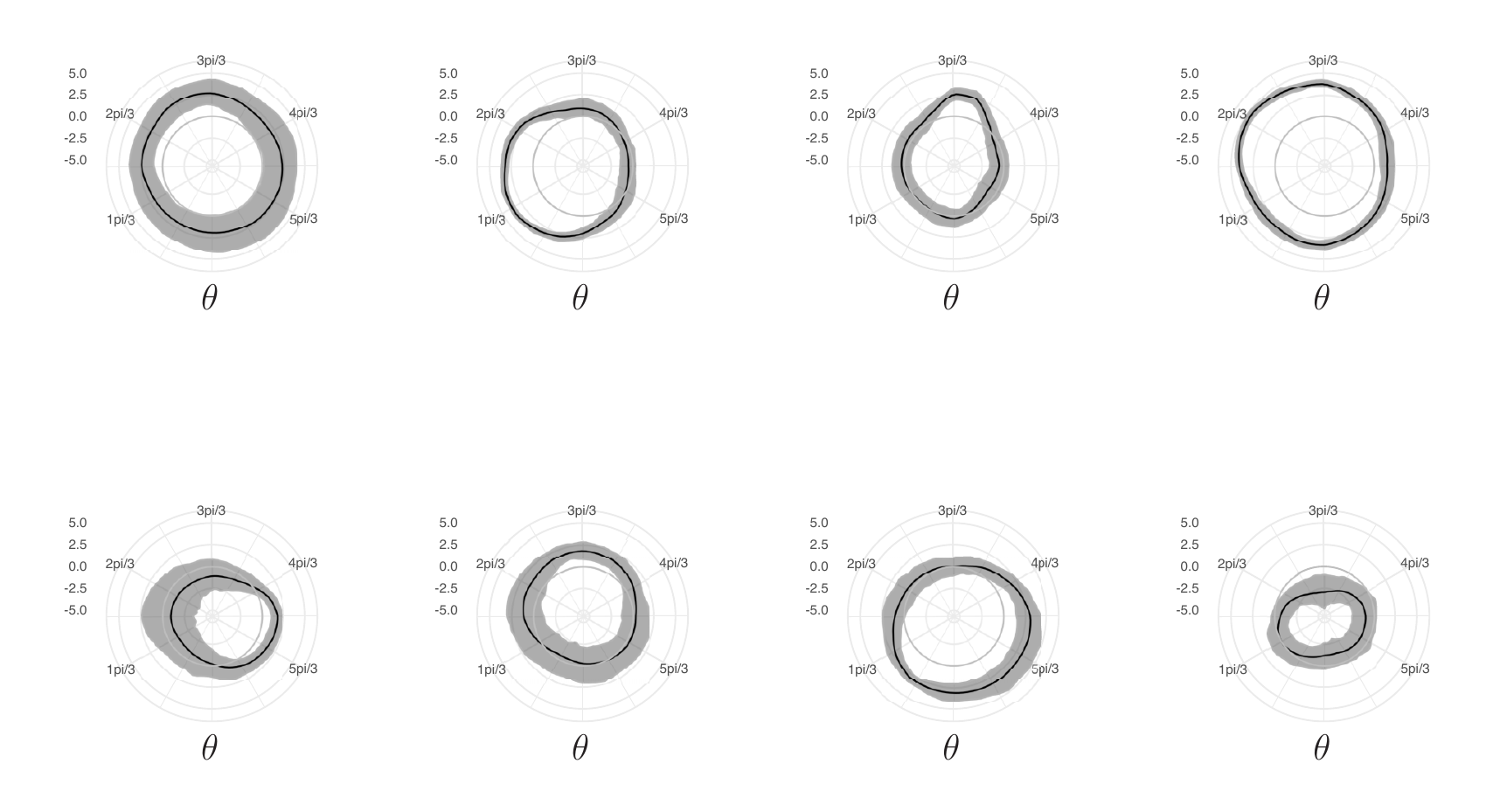}
  \caption{Posterior mean (solid black) and pointwise 95\% credible
    intervals (grey ribbon) of
    $\xi_{\Omega\times\Theta}\{(s_x, s_y), \theta\}$ from $\MA$, shown
    in clockwise order for a sequence of equally spaced values of
    $(s_x, s_y)$ on the diagonal ranging from $(s_x,s_y)=(0,0)$ to
    $(s_x,s_y)=(100,100)$.\ The circle shown in dark grey shows the
    constant function which is equal to zero.}
  \includegraphics{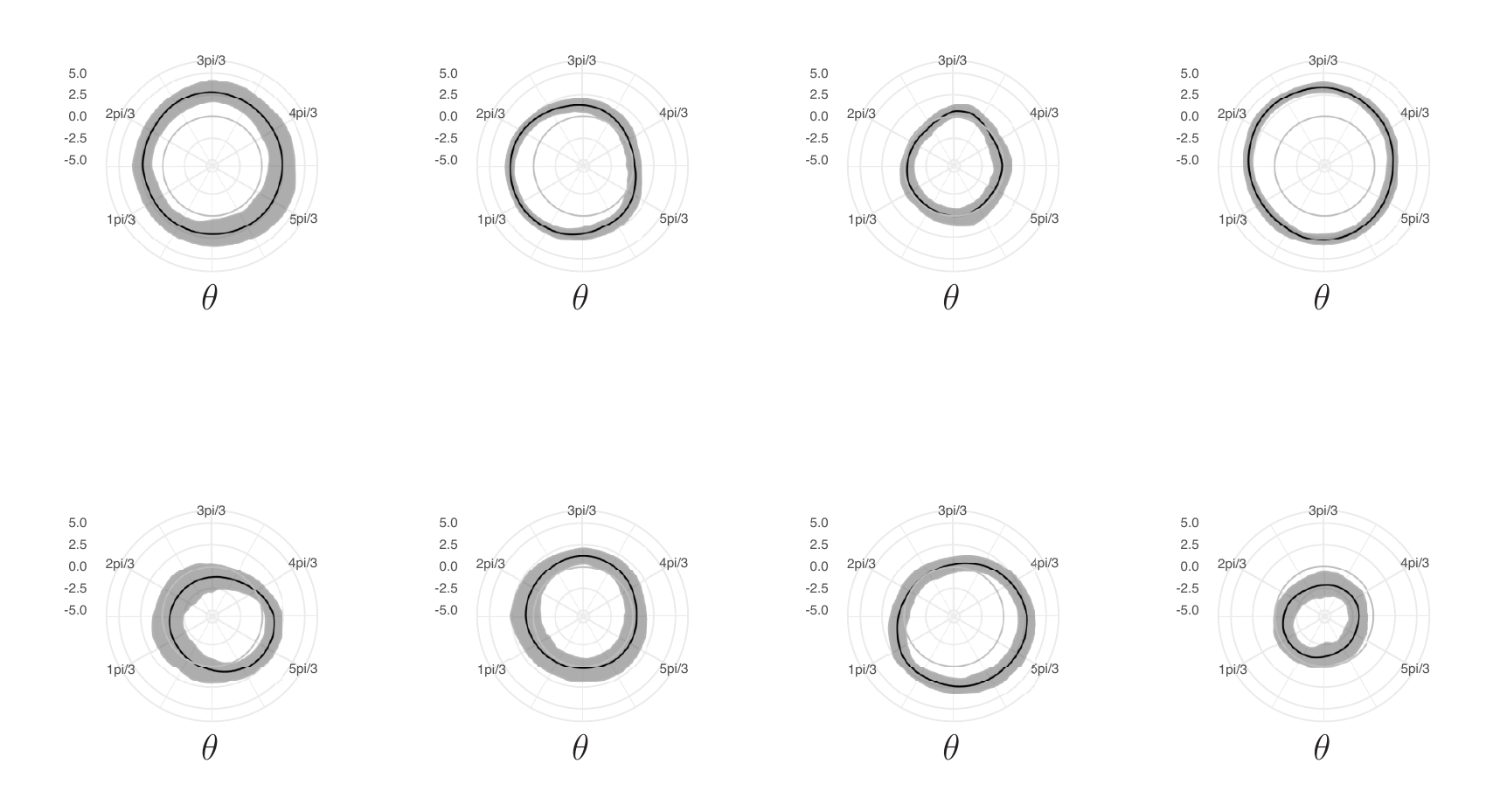}
  \caption{Posterior mean (solid black) and pointwise 95\% credible
    intervals (grey ribbon) of
    $\xi_{\Omega\times\Theta}\{(s_x, s_y), \theta\}$ from $\MB$, shown
    in clockwise order for a sequence of equally spaced values of
    $(s_x, s_y)$ on the diagonal ranging from $(s_x,s_y)=(0,0)$ to
    $(s_x,s_y)=(100,100)$.\ The circle shown in dark grey shows the
    constant function which is equal to zero.}
  \label{fig:xitheta}
\end{figure}

\begin{figure}[htpb!]
  \centering
  \includegraphics[scale=1]{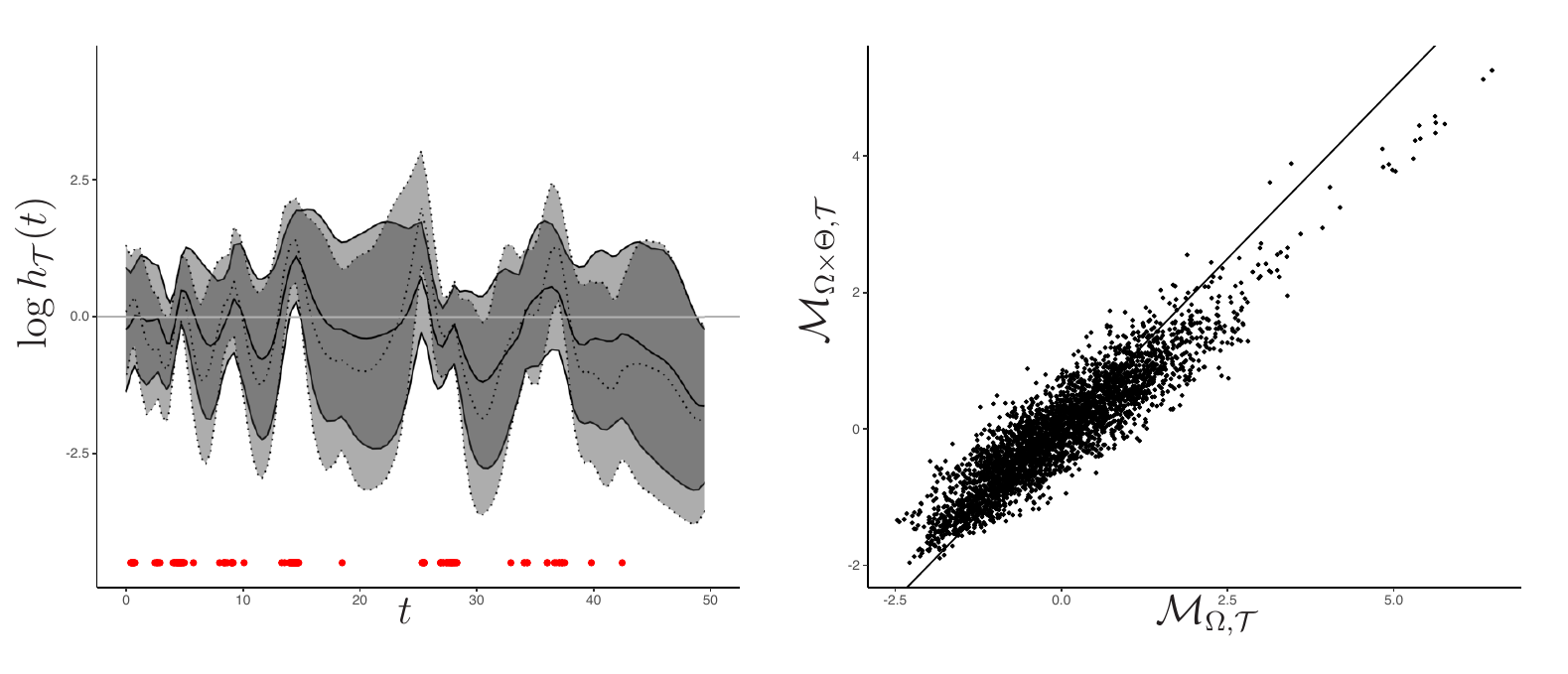}
  \caption{\textit{Left}: Posterior mean and pointwise 95\% credible
    intervals of $\xi_\TTT(t)$ from $\MOT$ (dotted) and $\MB$
    (solid) plotted over the time interval $0$--$50$
    seconds. \textit{Right}: scatterplot of posterior mean values of
    $\xi_\TTT(t)$ from $\MOT$ and $\MB$ at the nodes of the
    temporal mesh, which covers the entire interval of the
    experiment.}
  \label{fig:log_modulation_small}
\end{figure}

\section{Discussion}
\label{sec:discussion}

We introduce a novel and principled framework for analysis of neural
representations of space. We developed new statistical models, based
on Cox processes, that are amenable to modelling the relationship
between neural activity and covariates, including time, space and
other covariates such as head direction. Our approach rests on
modelling firing activity using Poisson point processes with latent
Gaussian effects. Our analyses in Sections 4.2 and 4.3 demonstrate
that our proposed models have better predictive performance than
currently widely used models
(e.g. \cite{sargolini2006conjunctive}. The latent prior Gaussian
effects accommodate persistent inhomogenous spatial-directional
patterns and overdispersion. Inference was performed in a fully
Bayesian manner, which allowed us to quantify uncertainty and to
identify effects that are typically missed out from most previous
analyses. Our approach supersedes a recently introduced framework that
uses Poisson generalised linear models with latent effects to model
grid cell activity (\cite{rule2023variational}): it does not involve
binning of data; incorporates more flexible priors that are
quasi-oscillatory; accommodates local distortions; by design does not
assume long-range dependence; and enables modelling of the effect of
covariates without the risk of ecological bias. While our focus here
is on the development and validation of the modelling framework, it
could be extended to generate metrics of grid field properties from
the model fitting hyper-parameters and statistical models for the
effects of experimental manipulations, e.g. effects of objects and
environmental geometry.

\small
\section*{Declaration} All animal procedures were performed under a UK
Home Office licence (PC198F2A0) in accordance with The University of
Edinburgh Animal Welfare committee's guidelines.  All procedures
complied with the Animals (Scientific Procedures) Act, 1986, and were
approved by a named Veterinary Surgeon and local ethical review
committee.

\section{Supplementary material}
\label{sec:supplement}

\subsection{$\mathsf{R}$ code and data set}
\begin{description}
\item[GitHub repository:] Repository containing \texttt{R} code to
  perform model fitting of the models described in the article:\newline
  \url{https://github.com/ipapasta/grid_fields}
  
\item[Grid cells data set:] Data set used in the analysis of our
  statistical methods in Section 4: \newline
  \url{https://datashare.ed.ac.uk/handle/10283/3674}
\end{description}

\subsection{Covariance of cyclic Mat\'{e}rn process}
\label{app:covariance_circular}
\begin{proof}[Proof of Proposition \ref{prop:covariance_circular}]
  For a regular mesh $0=k_0<\dots< k_{p-1} < k_{p}= 2\pi$ on $\Theta$
  with $k_i-k_{i-1} = h$, for all $i=1,\dots, p$, we have that
  $\kappa^2 C + G = \text{circulant}(\kappa^2 h+(2/h), -1/h, 0, \dots,
  0, -1/h)$. The elements of the inverse of a circulant matrix having
  only three non-zero elements in each row can be derived analytically
  from the solution of a recurrence equation \citep{sear79}, not
  least, we have that
  $(\kappa^2 C + G)^{-1} = \text{circulant}(a_0, a_1, \dots,
  a_{p-1})$, where
   \[
     a_j = \frac{z_1 z_2}{(-1/h)(z_1 -
       z_2)}\left(\frac{z_1^j}{1-z_1^{p}}-\frac{z_2^j}{1-z_2^{p}}\right),
     \qquad j=0,\dots,p-1
   \]
   with
   \[
     z_1 = \frac{-\kappa^2 h - (2/h)+ (\kappa^2 h^2 + 4
       \kappa^2)^{1/2}}{-2/h}\qquad\text{and} \qquad z_2 =
     \frac{\kappa^2 h + (2/h)+ (\kappa^2 h^2 + 4
       \kappa^2)^{1/2}}{2/h}.
   \]
   Let $\theta = 2\pi j/p\in\Theta$. Then, as $p \rightarrow \infty$
   and $h\rightarrow 0$, the covariance function
   $r(\theta) = a_{p\theta/(2\pi)}$ of the limiting process converges
   to
   \[
     r(\theta) \rightarrow \frac{1}{2\kappa}\frac{\cosh
       \{(\pi-\theta)\kappa\}}{\sinh (\pi \kappa)}\qquad\text{as
       $p\rightarrow\infty$},
   \]
   The case $\alpha=2$ is obtained by direct convolution of the
   covariance function associated with $\alpha=1$ \citep{lindetal11},
   that is,
   \[
     \frac{1}{(2\kappa)^2 \sinh(\pi \kappa)} \int_{0}^{2\pi}
     \cosh\{(\pi - \lvert u \lvert) \kappa\} \, \cosh\{(\pi - \lvert
     \theta-u \lvert) \kappa\}~\dd u =: \zeta_{\Theta}(\theta).
   \]
 \end{proof}
 \subsection{Variances of Gaussian random fields}
 \label{sec:variances}
 \begin{proof}[Proof of Proposition~\ref{prop:variance-on-R2}]
   Explicit expressions for the variance for the case where
   $\phi_\Omega \in (-1,1]$ are derived and documented in
   \cite{lindetal11}. Thus, we proceed by deriving the variance of the
   stationary Gaussian random field for the overdamped case, that is,
   when $\phi_\Omega > 1$.\ Dropping the manifold index from the
   parameters, we have that for $\phi > 1$, direct computation of the
   variance from the spectral density, after a change from Cartesian
   to polar coordinates \citep{lindgren2012stationary}, is equal to to
   \begin{IEEEeqnarray*}{rCl}
     \pVar\{\xi_\Omega (s)\}&=&(2\pi)\int_0^\infty \frac{r}{(2\pi)^2}
     \frac{\sigma^2}{\kappa^4 + 2\kappa^2 \phi r^2 + r^4} \dd r =
     \frac{\sigma^2}{4\pi} \int_0^\infty
     \frac{\dd u}{\kappa^4 + 2\kappa^2 \phi u + u^2}=\\ \\
     &&= \frac{\sigma^2}{4\pi} \int_0^\infty \frac{\dd
       u}{(u+\kappa^2\phi)^2 + \kappa^4 - \kappa^4\phi^2} =
     \frac{\sigma^2}{4\pi}
     \int\limits_{\frac{\phi}{\sqrt{\phi^2-1}}}^\infty
     \frac{\dd v}{(\kappa^4\phi^2 - \kappa^2)v^2 - (\kappa^4 \phi^2 - \kappa^2)}=\\ \\
     &&=\frac{\sigma^2}{(4\pi)\sqrt{\kappa^4\phi^2-\kappa^4}}
     \int\limits_{\frac{\phi}{\sqrt{\phi^2-1}}}^\infty (v^2 -
     1)^{-1}\dd v=
     \\\\
     &&= \frac{\sigma^2}{(4\pi)\sqrt{\kappa^4\phi^2-\kappa^4}}
     \left[\frac{\iota
         \pi}{2}\lim_{v\to\infty}\{\log(1+v)-\log(v-1)\}+\frac{1}{2}\log\left(1+\frac{\phi}{\sqrt{\phi^2-1}}\right)-\right.\\\\
     &&\left. \qquad \qquad \qquad \qquad \qquad -\frac{\iota
         \pi}{2}-\frac{1}{2}\log\left(\frac{\phi}{\sqrt{\phi^2-1}}-1\right)\right]\\\\
     &=& \frac{\sigma^2}{4\pi\kappa^2
       \sqrt{\varphi^2-1}}\left\{\frac{\pi}{2} - \text{atan}\{
       (\varphi^2-1)^{-1/2}\,{\varphi}\}\right\}
   \end{IEEEeqnarray*}
 \end{proof}

 \begin{proof}[Proof of Proposition \ref{prop:variance-on-S1}]
  Dropping the manifold index from the parameters, direct calculation
  of the variance from the spectral mass function gives that
\begin{equation}
  \pVar\{\xi_\Theta(\theta)\} = \sum_{\omega \in \ZZ}
  f_{\theta}(\omega)=
  \frac{\sigma ^2 \pi \kappa}{4 \kappa^3 \left(\varphi^2-1\right)^{1/2}}
  \left(\frac{\cot \left(i\, z\right)} {i \,z}-\frac{\cot \left(i\,
        \overline{z} \right)} {i \,\overline{z}} \right),\quad \varphi \in
  (-1,\infty),
  \label{eq:var_theta}
\end{equation}
where
$\overline{z}= \pi \kappa \left(\varphi + (1-\varphi ^2)^{1/2}\,i
\right)^{1/2}$.\ When $\varphi \in [1,\infty)$, $i\,z \in \RR$ is real
and expression \eqref{eq:var_theta} gives the variance that is stated
in the proposition.

For $\varphi \in (-1,1)$ we use the trigonometric
identity $\cot(i \, A)/i\, A = -\coth(A)/A$ for any $A\in
\mathbb{C}$, to rewrite expression \eqref{eq:var_theta} as
\[
  \pVar\{\xi_\Theta(\theta)\} = \frac{\sigma ^2 \pi \kappa}{4 \kappa^3
    \left(1-\varphi^2\right)^{1/2}\,i} \left(\frac{\coth \left(
        \overline{z}\right)} {\overline{z}}-\frac{\coth \left( z
      \right)} {z} \right),
\]
which, upon further simplification based on the identity
\[
  \frac{\coth (x+i y)}{x+i y}-\frac{\coth (x-i y)}{x-i y} \frac{2 \,
    i\, }{\left(x^2+y^2\right) } = \frac{x \sin (2 y)+y \sinh (2 x)}{\cos
    (2 y)-\cosh (2 x)}, \qquad \text{$x,y\in\RR$},
\]
proves the claim for the underdamped case. 
\end{proof}

\begin{proof}[Proof of Proposition \ref{prop:variance-on-R}]
  Dropping the manifold index from the parameters, for $-1<\varphi<1$,
  $\pVar\{\xi_\TTT (t)\} $ is equal to 
\begin{IEEEeqnarray*}{rCl}
  && \int_{-\infty}^\infty S(\omega;\kappa,\varphi) \md\omega =
  \frac{1}{\pi}\int_0^\infty
  \frac{1}{\kappa^4+2\varphi\kappa^2\omega^2+\omega^4} \md\omega \\\\
  &=&\frac{1}{\pi}\int_0^\infty
  \frac{\kappa}{2\sqrt{u}}\frac{\md u}{\kappa^4+2\varphi\kappa^4
    u+\kappa^4 u^2}  =\frac{1}{2\pi\kappa^3}\int_0^\infty
  \frac{\md u}{\sqrt{u}}\frac{1}{1+2\varphi u+ u^2}  \\\\
  &=&\frac{1}{2\pi\kappa^3} 2^{1/2}
  (1-\varphi^2)^{-1/4}\Gamma(3/2)B(1/2,3/2)P^{-1/2}_{-1}(\varphi)
\end{IEEEeqnarray*}
where the last expression follows after using \citet[expression
3.252.10 in][]{gradshteyn2014table} with $\mu=1/2$, $\nu=1$, and
$\cos(t)=\varphi$, so that $\sin(t)=\sqrt{1-\varphi^2}$.  \if0\blind{
  where in the third equality, we used changed variables according to
  $u=\omega^2/\kappa^2$.  } \fi Note that $\Gamma(3/2)=\sqrt{\pi}/2$
and $B(1/2,3/2)=\pi/2$. Hence, using the sin-half-angle formula
\citep[expression 8.754 in][]{gradshteyn2014table} we have that
\begin{IEEEeqnarray*}{rCl}
  \pVar\{\xi_\TTT(t)\} &=& \frac{\pi^{1/2}}{2^{5/2}\kappa^3}
  (1-\varphi^2)^{-1/4} P^{-1/2}_{-1}(\varphi)=
  \frac{\pi^{1/2}}{2^{5/2}\kappa^3} (1-\varphi^2)^{-1/4}
  \sqrt{\frac{2}{\pi\sqrt{1-\varphi^2}}}
  \frac{\sqrt{(1-\varphi)/2}}{1/2} \\ &=& \frac{2}{2^{5/2}\kappa^3}
  \sqrt{\frac{1-\varphi}{1-\varphi^2}} = \frac{1}{4\kappa^3}
  \sqrt{\frac{2}{1+\varphi}}.
\end{IEEEeqnarray*}

Now, consider the case $\varphi>1$. We have
\begin{IEEEeqnarray*}{rCl}
  \pVar\{\xi_\TTT(t)\} &=& 
  \frac{1}{2\pi}\int_{-\infty}^\infty
  \frac{\md\omega}{\kappa^4+2\varphi\kappa^2\omega^2+\omega^4}  =
  \frac{\kappa}{2\pi}\int_0^\infty
  \frac{\md u}{\kappa^4+2\varphi\kappa^4 u^2+\kappa^4 u^4} 
  \\&=&\frac{1}{2\pi\kappa^3}\int_{-\infty}^\infty \frac{\md u}{1+2\varphi
    u^2+ u^4} = \frac{1}{2\pi\kappa^3}\int_{-\infty}^\infty
  \frac{\md u}{2\sqrt{\varphi^2-1}}\left(
    \frac{1}{u^2+\varphi-\sqrt{\varphi^2-1}} -
    \frac{1}{u^2+\varphi+\sqrt{\varphi^2-1}} \right) 
\end{IEEEeqnarray*}
\if0\blind{in third equatity, change variables according to
  $u=\omega/\kappa\right$ }\fi
where the last step is obtained via the
polynomial roots $a_{\pm}=-\varphi\pm\sqrt{\varphi^2-1}<0$ (with
respect to $u^2$).

Next, we have that $\pVar\{\xi_\TTT(t)\}$ is equal to
\begin{IEEEeqnarray*}{rCl}
&&\frac{1}{4\pi\kappa^3\sqrt{\varphi^2-1}} \left(
\int_{-\infty}^\infty \frac{\md u}{u^2-a_+}  - \int_{-\infty}^\infty
\frac{\md u}{u^2-a_-} \right) =
\frac{1}{4\pi\kappa^3\sqrt{\varphi^2-1}} \left( \int_{-\infty}^\infty
\frac{\sqrt{-a_{+}} \md v}{-v^2a_{+}-a_{+}}  - \int_{-\infty}^\infty
\frac{\sqrt{-a_{-}}\md v}{-v^2a_{-} -a_{-}}  \right) \\\\&=&
\frac{1}{4\pi\kappa^3\sqrt{\varphi^2-1}} \left(
\frac{\sqrt{-a_+}}{-a_+} - \frac{\sqrt{-a_-}}{-a_-} \right)
\int_{-\infty}^\infty \frac{\md v}{v^2 + 1}  =
\frac{1}{4\pi\kappa^3\sqrt{\varphi^2-1}} \left( \frac{1}{\sqrt{-a_+}}
- \frac{1}{\sqrt{-a_-}} \right) \left[ \arctan(v)
\right]_{v=-\infty}^\infty \\&=&
\frac{1}{4\kappa^3\sqrt{\varphi+1}\sqrt{\varphi-1}} \left(
\frac{1}{\sqrt{\varphi-\sqrt{\varphi^2-1}}} -
\frac{1}{\sqrt{\varphi+\sqrt{\varphi^2-1}}} \right).
\end{IEEEeqnarray*}
\if0\blind{ in third equality
  $\left[{v=u/\sqrt{-a_{+}}\text{ and }v=u/\sqrt{-a_{-}}}\right]$ }\fi
Let $b_+=-a_-$ and $b_-=-a_+$.\ Then, using that
$b_-+b_+=\varphi-\sqrt{\varphi^2-1}+\varphi+\sqrt{\varphi^2-1}=2\varphi$
and
$b_-b_+=(\varphi-\sqrt{\varphi^2-1})(\varphi+\sqrt{\varphi^2-1})=\varphi^2-\varphi^2+1=1$,
\begin{IEEEeqnarray*}{rCl}
  \pVar\{\xi_\TTT(t)\} &=&
  \frac{1}{4\kappa^3\sqrt{\varphi+1}}\sqrt{\frac{2}{b_++b_--2\sqrt{b_-b_+}}}
  \left( \frac{1}{\sqrt{b_-}} - \frac{1}{\sqrt{b_+}} \right) \\\\&=&
  \frac{1}{4\kappa^3}\sqrt{\frac{2}{\varphi+1}}
  \sqrt{\frac{1}{(\sqrt{b_+}-\sqrt{b_-})^2}}
  \frac{\sqrt{b_+}-\sqrt{b_-}}{\sqrt{b_+b_-}} =
  \frac{1}{4\kappa^3}\sqrt{\frac{2}{\varphi+1}} .
\end{IEEEeqnarray*}
\end{proof}

\subsection{Additional Tables and Figures}
\label{sec:table_and_figures}
In this section, we provide additional Tables and Figures that are
referenced in the main text. Table~\ref{tab:score-negative-proportion}
shows the proportions of negative score differences between models
$\mathcal{M}$ and $\MO$.
\begin{table}[htbp!]
\centering
\begin{tabular}{r|rrr|rrr}
  \hline
  &\multicolumn{3}{c|} {SE}& \multicolumn{3}{c} {DS}\\
  \hline\hline
  interval & \MA & \MOT & \MB  &\MA  & \MOT & \MB  \\
  \hline
  \multirow{6}{*}{Fold 1$\quad$}
  2s & 0.61 & 0.57 & 0.58 & 0.60 & 0.52 & 0.55 \\ 
  5s & 0.53 & 0.44 & 0.51 & 0.56 & 0.46 & 0.50 \\ 
  10s& 0.57 & 0.30 & 0.33 & 0.66 & 0.58 & 0.62 \\ 
  20s& 0.64 & 0.25 & 0.48 & 0.64 & 0.64 & 0.66 \\ 
  30s& 0.69 & 0.52 & 0.59 & 0.76 & 0.76 & 0.83 \\ 
  40s& 0.73 & 0.18 & 0.50 & 0.73 & 0.68 & 0.68 \\
  \hline
  \multirow{6}{*}{Fold 2$\quad$}
  2s & 0.58 & 0.51 & 0.58 & 0.59 & 0.54 & 0.57 \\ 
  5s & 0.56 & 0.47 & 0.50 & 0.61 & 0.53 & 0.56 \\ 
  10s& 0.57 & 0.39 & 0.54 & 0.61 & 0.51 & 0.56 \\ 
  20s& 0.69 & 0.40 & 0.51 & 0.76 & 0.64 & 0.67 \\ 
  30s& 0.57 & 0.20 & 0.40 & 0.60 & 0.50 & 0.53 \\ 
  40s& 0.65 & 0.30 & 0.48 & 0.83 & 0.78 & 0.83 \\
  \hline
  \multirow{6}{*}{Folds combined$\quad$}  
  2s & 0.59 & 0.54 & 0.58 & 0.59 & 0.53 & 0.56 \\ 
  5s & 0.55 & 0.45 & 0.51 & 0.59 & 0.50 & 0.53 \\ 
  10s& 0.57 & 0.34 & 0.44 & 0.63 & 0.54 & 0.59 \\ 
  20s& 0.66 & 0.33 & 0.49 & 0.70 & 0.64 & 0.66 \\ 
  30s& 0.63 & 0.36 & 0.49 & 0.68 & 0.63 & 0.68 \\ 
  40s& 0.69 & 0.24 & 0.49 & 0.78 & 0.73 & 0.76 
\end{tabular}
\caption{Proportion of negative squared-error (SE) and
  Dawid--Sebastiani (DS) score differences between model $\mathcal{M}$
  and $\MO$, for $\mathcal{M}\in\{\MA, \MOT,\MB\}$.}
\label{tab:score-negative-proportion}
\end{table}
Figure~\ref{tab:score-stddev} shows conservative estimates of the
standard deviation of squared-error score differences between model
$\mathcal{M}$ and model $\MO$.
\begin{table}[htbp!]
  \centering
  \begin{tabular}{r||lll|lll}
    \hline 
    &\multicolumn{3}{c|} {SE}& \multicolumn{3}{c} {DS}\\
    \hline\hline
    interval & \MA & \MOT & \MB  &\MA  & \MOT & \MB  \\ 
    \hline\hline
    \multirow{6}{*}{Fold 1$\quad$}
    2s & 160 & 903 & 1886 & 39.6 & 52.4 & 51.9 \\           
    5s & 454 & 8339 & 8502 & 31.0 & 58.5 & 57.8 \\          
    10s& 1168 & 11147 & 17520 & 80.4 & 97.4 & 98.2 \\       
    20s& 1594 & 12504 & 12457 & 19.8 & 48.7 & 48.0 \\       
    30s& 4926 & 6444 & 3655 & 26.8 & 41.6 & 40.9 \\         
    40s& 2582 & 10903 & 3744 & 9.6 & 11.7 & 11.7 \\
    \hline
    \multirow{6}{*}{Fold 2$\quad$}
    2s & 238 & 1828 & 1846 & 54.7 & 63.2 & 65.8 \\          
    5s & 452 & 2490 & 2822 & 192.5 & 247.9 & 248.7 \\       
    10s& 793 & 4621 & 4513 & 76.2 & 81.5 & 69.6 \\          
    20s& 3423 & 10408 & 15885 & 113.9 & 583.5 & 551.6 \\    
    30s& 1728 & 12957 & 6586 & 19.7 & 25.6 & 25.3 \\        
    40s& 5292 & 16016 & 6953 & 48.3 & 56.1 & 55.8 \\
    \hline
    \multirow{6}{*}{Average$\quad$}
    2s & 199 & 1366 & 1866 & 47.2 & 57.8 & 58.8 \\          
    5s & 453 & 5414 & 5662 & 111.8 & 153.2 & 153.3 \\       
    10s& 980 & 7884 & 11016 & 78.3 & 89.5 & 83.9 \\         
    20s& 2508 & 11456 & 14171 & 66.9 & 316.1 & 299.8 \\     
    30s& 3327 & 9700 & 5120 & 23.3 & 33.6 & 33.1 \\         
    40s& 3937 & 13460 & 5348 & 29.0 & 33.9 & 33.7      
  \end{tabular}
  \caption{Conservative estimates of standard deviation of
    squared-error (SE) and Dawid--Sebastiani (DS) score differences
    between model $\mathcal{M}$ and model $\MO$ with
    $\mathcal{M}$$\in\{\MA,
    \MOT,\MB\}$, rounded to the nearest integer.}
    \label{tab:score-stddev}
  \end{table}
  Figures \ref{fig:xiomega_CI_MA_low}, \ref{fig:xiomega_CI_MB_low},
  \ref{fig:xiomega_CI_MA_upp} and \ref{fig:xiomega_CI_MB_upp} show
  pointwise 95\% credible intervals of
  $\xi_{\Omega\times\Theta}(s, \theta)$ for models $\MA$ and $\MB$.
  \begin{figure}[htpb!]
    \centering
    \includegraphics[scale=1]{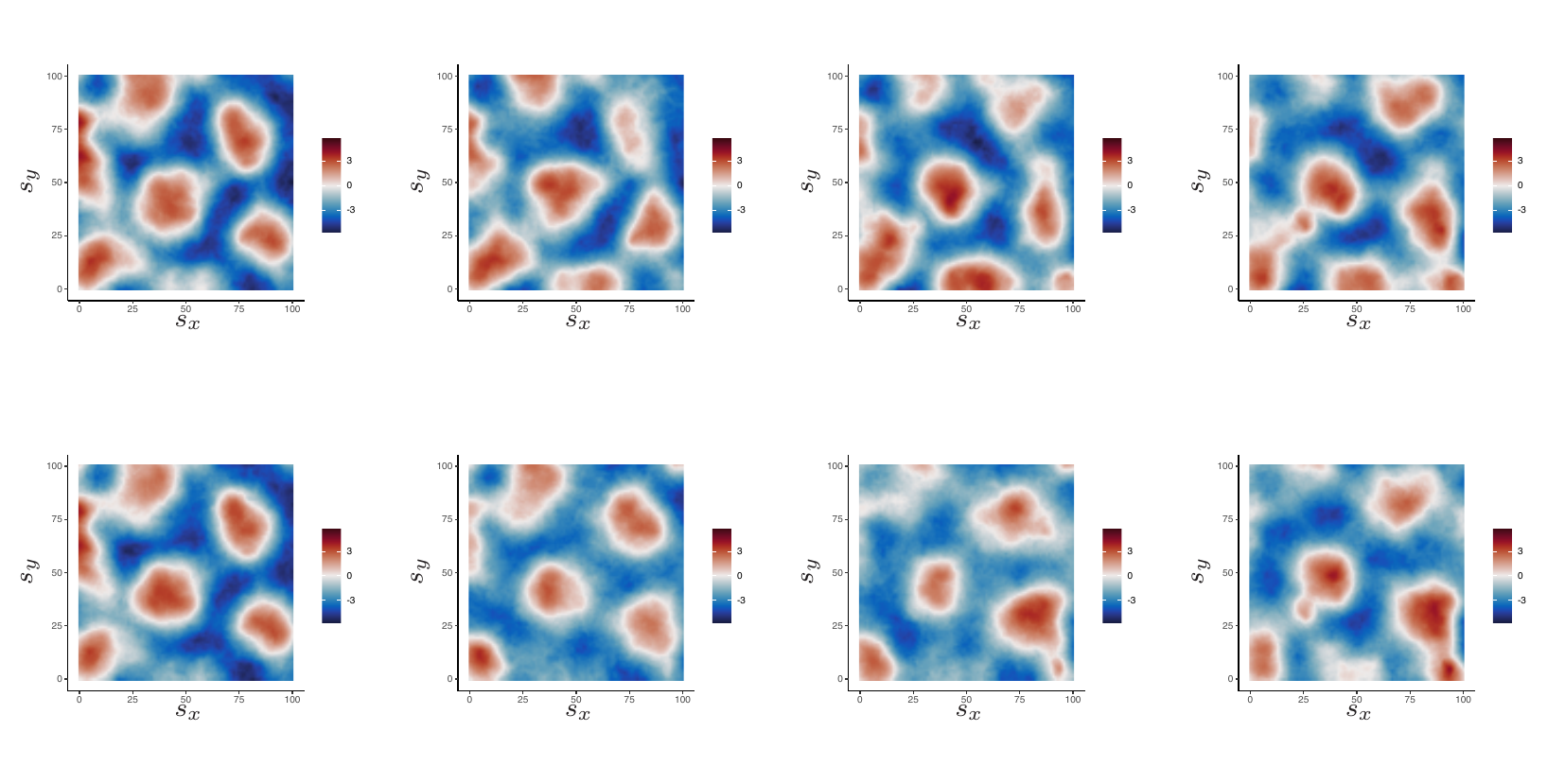}
    \caption{Posterior 0.025 quantile of
      $\xi_{\Omega\times\Theta}((s_x, s_y), \theta)$ from $\MA$, shown in clockwise
      order for a sequence of equally spaced $\theta$ values ranging
      from $\pi/8$ (top left) to $2\pi$ (bottom left).}
    \label{fig:xiomega_CI_MA_low}
    \includegraphics[scale=1]{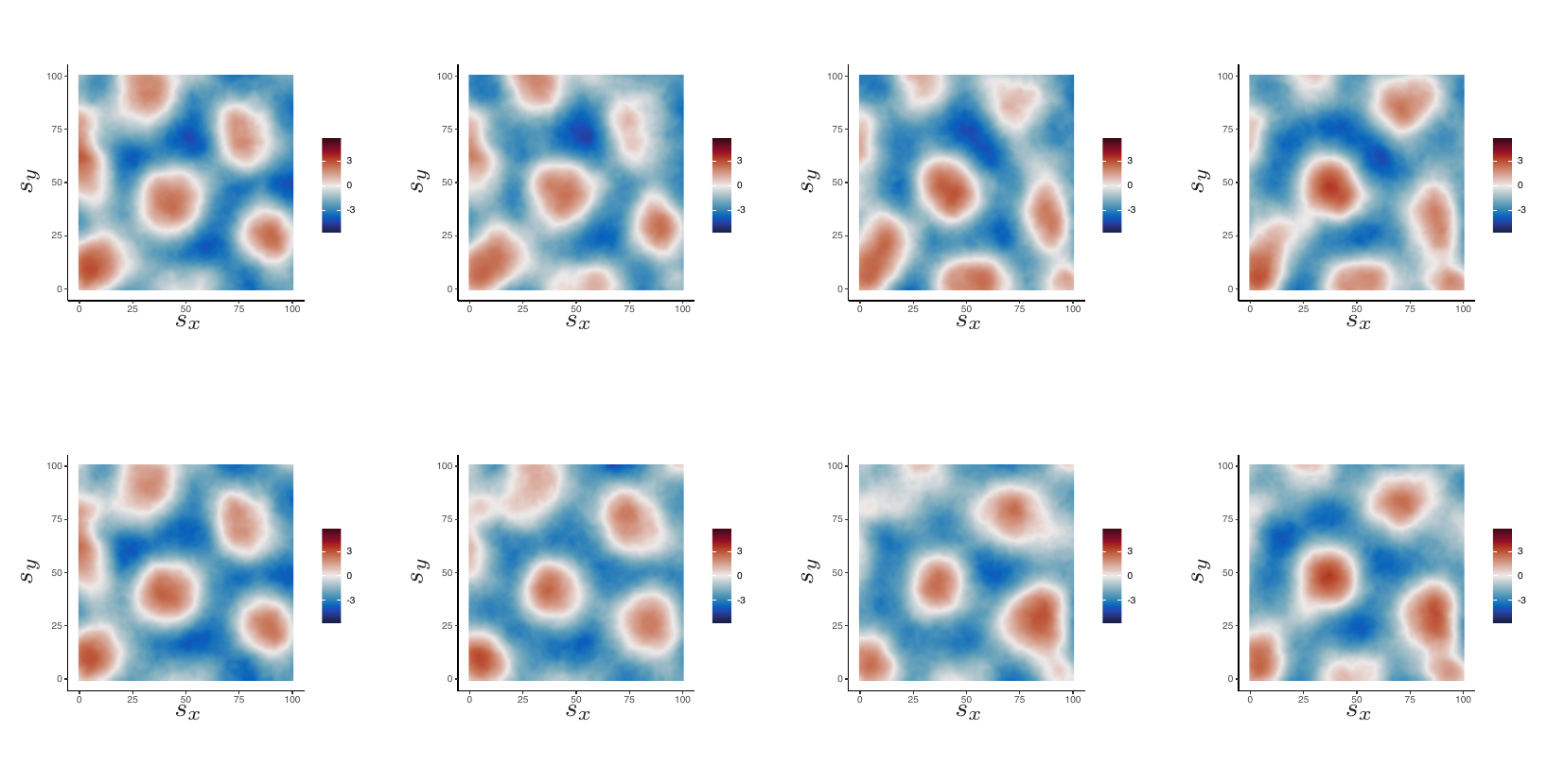}
    \caption{Posterior 0.025 quantile of
      $\xi_{\Omega\times\Theta}((s_x, s_y), \theta)$ from $\MB$, shown in clockwise
      order for a sequence of equally spaced $\theta$ values ranging
      from $\pi/8$ (top left) to $2\pi$ (bottom left).}
    \label{fig:xiomega_CI_MB_low}
  \end{figure}

  \begin{figure}[htpb!]
    \centering
    \includegraphics[scale=1]{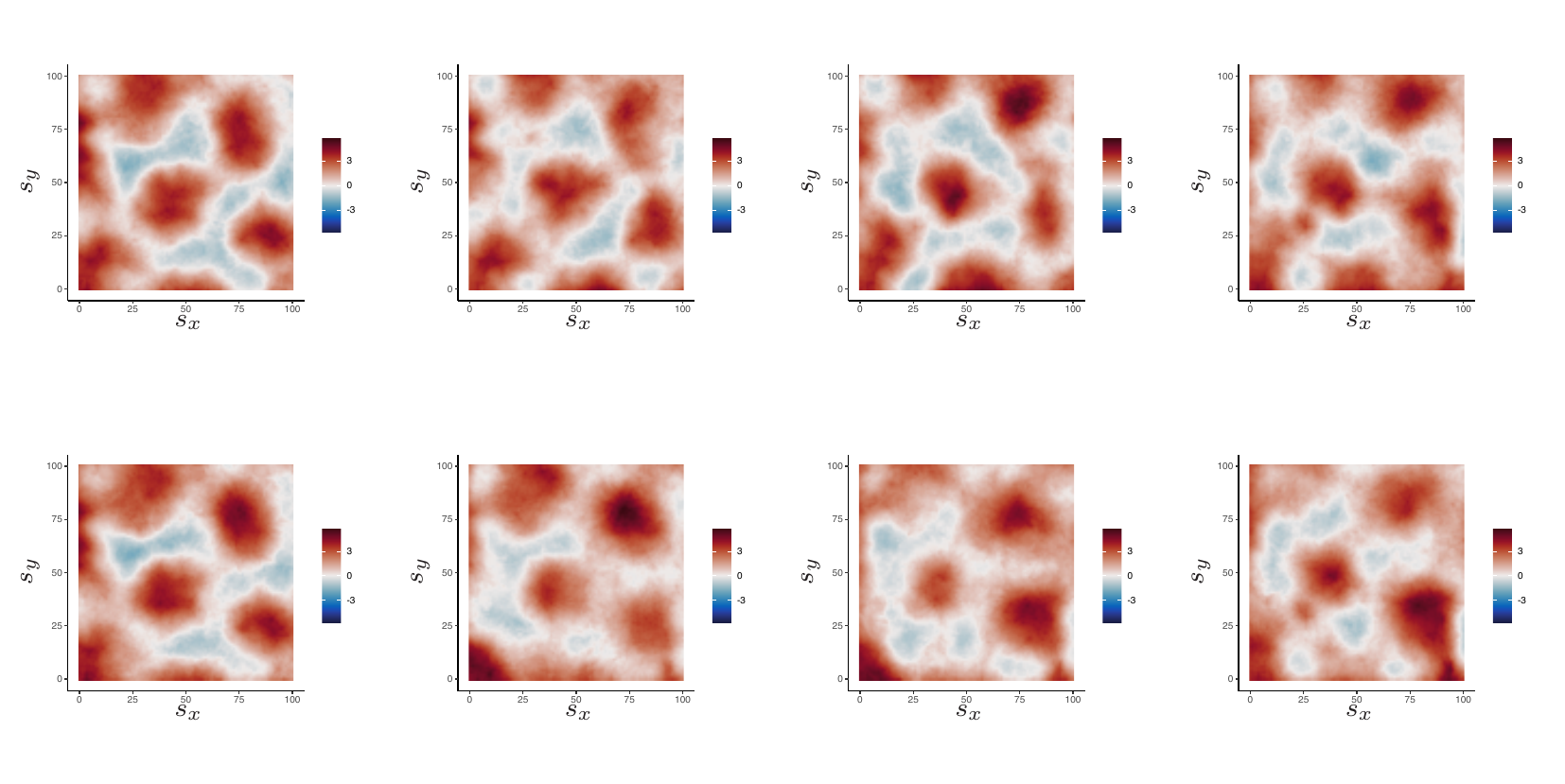}
    \caption{Posterior 0.975 quantile of
      $\xi_{\Omega\times\Theta}((s_x, s_y), \theta)$ from $\MA$, shown in clockwise
      order for a sequence of equally spaced $\theta$ values ranging
      from $\pi/8$ (top left) to $2\pi$ (bottom left).}
        \label{fig:xiomega_CI_MA_upp}
    \includegraphics[scale=1]{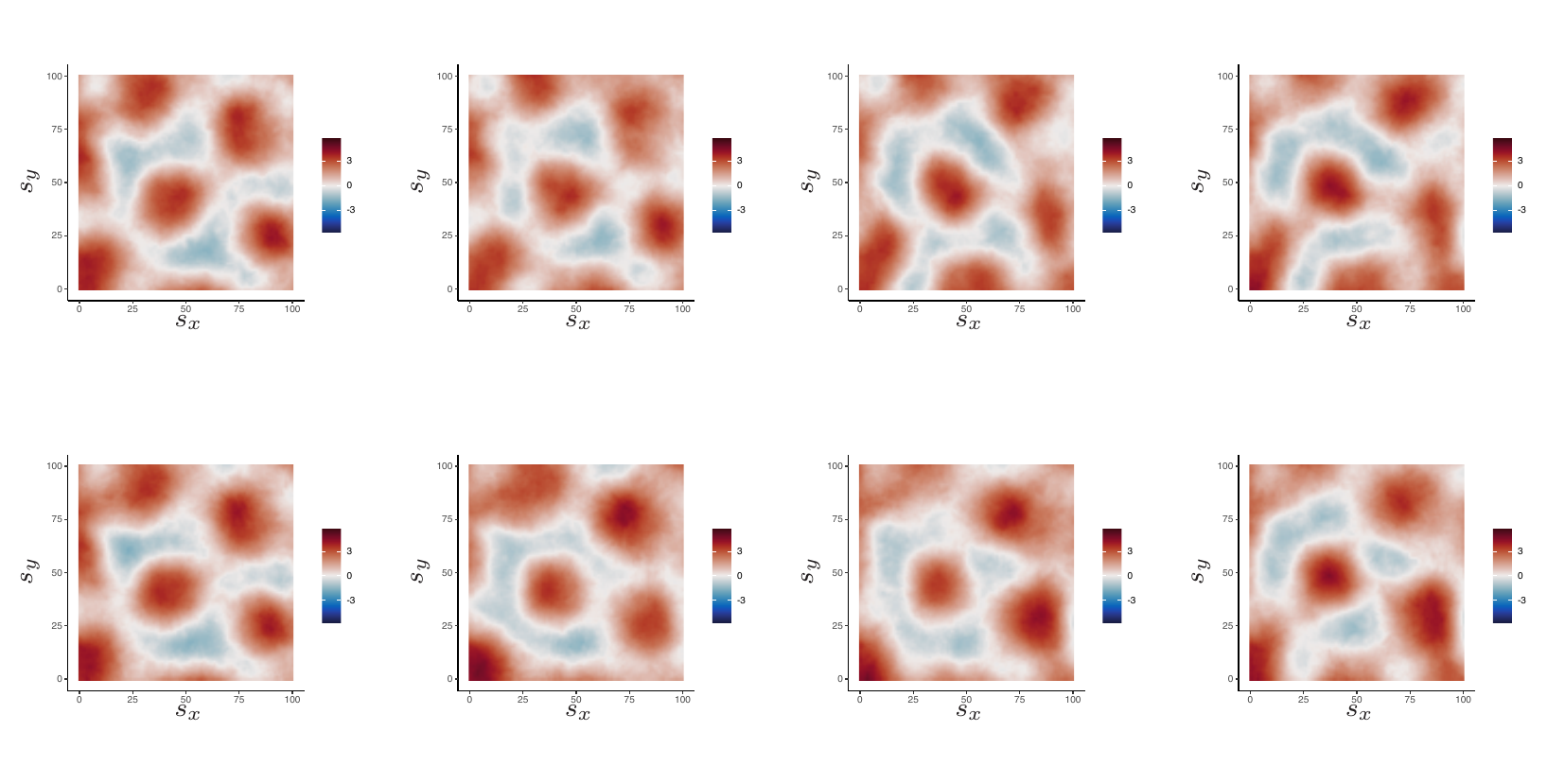}
    \caption{Posterior 0.975 quantile of
      $\xi_{\Omega\times\Theta}((s_x, s_y), \theta)$ from $\MB$, shown
      in clockwise order for a sequence of equally spaced $\theta$
      values ranging from $\pi/8$ (top left) to $2\pi$ (bottom left).}
    \label{fig:xiomega_CI_MB_upp}
  \end{figure}

  \subsection{Permutation test}
  \label{sec:perm-test}
\begin{algorithm}
\SetKwData{Left}{left}\SetKwData{This}{this}\SetKwData{Up}{up}
\SetKwFunction{Union}{Union}\SetKwFunction{FindCompress}{FindCompress}
\SetKwInOut{Input}{input}\SetKwInOut{Output}{output}

\Input{$J\in\NN$ and negatively oriented model scores
  $S(F_i^{\mathcal{M}}, n_i), S(F_i^{\MO}, n_i)$,
  $1 \leq i \leq N_{\text{val}}$.}  \Output{$p$, an estimate of the
  $p$-value of the test.}  \For{$i\leftarrow 1$ \KwTo $N_{\text{val}}$}{
  \emph{compute the score difference}\;
  $S_i^- = S(F_i^{\mathcal{M}}, n_i) - S(F_i^{\MO}, n_i) $ \; }
\emph{compute the observed test statistic
  $T_{obs} = \frac{1}{N_{\text{val}}}\sum_{i=1}^{N_{\text{val}}}S_i^-
  $} \; \For{$j\leftarrow 1$ \KwTo $J$}{ \For{$i\leftarrow 1$ \KwTo
    $N_{\text{val}}$}{ \emph{compute the randomized score
      difference}\;
\begin{equation*}
  S_i^-(j) =
  \begin{cases} S_i^- \,\,\, & \text{with probability $0.5$}
    \\ -S_i^- \,\,\, & \text{with probability $0.5$} \\
  \end{cases}
\end{equation*}
} \emph{compute}
$T_j = \frac{1}{N_{\text{val}}}\sum_{i=1}^{N_{\text{val}}}S_i^-(j) $
\; } \emph{return}
$p = \frac{1}{J}\sum_{j=1}^{J} \mathbbm{1}[T_j \leq T_{obs}].$
\caption{Hypothesis testing procedure for testing pairwise
  exchangeability of model scores.  The computed $p$-value is close to
  zero if model $\mathcal{M}$ is better than model $\MO$, and close to
  $1$ if model $\MO$ is better.\ A two-sided test $p$-value can be
  obtained as $2\min(p,1-p)$.\ The positive integer $N_{\text{val}}$
  corresponds to the number of connected segments on the mouses
  trajectory in the validation set.}\label{TestAlg}
\end{algorithm}

\spacingset{0}

\end{document}